\documentclass[letter,11pt]{article}
\pdfoutput=1
\usepackage{amsmath,amssymb, amsthm}
\usepackage{complexity}
\usepackage{float}
\usepackage{algorithm}
\floatname{algorithm}{Protocol}
\usepackage{multirow}
\usepackage{graphicx}
\usepackage{amssymb,amsmath}
\usepackage[margin=1in]{geometry}
\usepackage[normalem]{ulem}
\usepackage{mathtools}

\newcounter{linenum}

\newcommand{\ket}[1]{\left\lvert #1 \right\rangle}
\newcommand{\bra}[1]{\left\langle #1 \right\lvert}

\usepackage{color}

\newtheorem{theorem}{Theorem}
\newtheorem{lemma}{Lemma}
\newtheorem{corollary}{Corollary}

\newtheorem{definition}{Definition}

\usepackage[affil-it]{authblk}
\begin{document}
\allowdisplaybreaks[3]
\frenchspacing

\hyphenation{avenues}
\hyphenation{research}
\hyphenation{interact}
\hyphenation{machine}
\hyphenation{machines}
\hyphenation{giving}
\hyphenation{paper}
\hyphenation{encoun-ter}
\hyphenation{encoun-ters}
\hyphenation{expe-ri-ence}
\hyphenation{analyse}
\hyphenation{analysis}
\hyphenation{remains}
\hyphenation{logical}
\hyphenation{cannot}

\clubpenalty 10000
\widowpenalty 10000

\title{Verified Delegated Quantum Computing with One Pure Qubit}

\author[1]{Theodoros Kapourniotis}
\author[1]{Elham Kashefi}
\author[2]{Animesh Datta}
\affil[1]{School of Informatics, University of Edinburgh\\
  10 Crichton Street, EH8 9AB, UK}
\affil[2]{Clarendon Laboratory, Department of Physics\\
University of Oxford, OX1 3PU, UK}

\maketitle

\begin{abstract}
While building a universal quantum computer remains challenging, devices of restricted power such as the so-called \textit{one pure qubit} model have attracted considerable attention. An important step in the construction of these limited quantum computational devices is the understanding of whether the verification of the computation within these models could be also performed in the restricted scheme. Encoding via blindness (a cryptographic protocol for delegated computing) has proven successful for the verification of universal quantum computation with a restricted verifier. In this paper, we present the adaptation of this approach to the one pure qubit model, and present the first feasible scheme for the verification of delegated one pure qubit model of quantum computing.
 \end{abstract}

\section{Introduction}

The physical realisation of quantum information processing requires the fulfilment of the five criteria collated by DiVincenzo \cite{divincenzo00}. While enormous progress had been made in realising them since, we are still some way from constructing a universal quantum computer. This raises the question whether quantum advantages in computation are possible without fulfilling one or more of DiVincenzo's criteria. From a more foundational perspective, the computational power of the intermediate models of computation are of great value and interest in understanding the computational complexity of physical systems. Several such models are known, including fermionic quantum computation \cite{BravyiKitaev02}, instantaneous quantum computation \cite{BJS11}, permutational quantum computation \cite{jordan10}, and boson sampling \cite{AA11}.

Deeply entwined with the construction of a quantum information processor is the issue of its verification. How do we convince ourselves that the output of a certain computation is correct and obtained using quantum-enhanced means. Depending on a given computation, one or both may be non-trivial. For instance, the correctness of the output of Shor's factoring algorithm~\cite{Shor} can be checked efficiently on a classical machine, but in general this is not known to be possible for all problems solvable by a quantum computer.
On the other hand, by allowing a small degree of quantumness to the verifier~\cite{Dorit,fitzsimons2012unconditionally}, or considering entangled non-commuting provers~\cite{RUV13}, the verification problem has been solved for universal quantum computation. However, not much attention has been given to verifying restricted models of quantum-enhanced computation. It is in this direction that we endeavour to embark.

One of the earliest restricted models of quantum computation was proposed by Knill and Laflamme, named `Deterministic Quantum Computation with One quantum bit (DQC1)', also referred to as the one pure qubit model~\cite{kl98}. It addresses the challenge of DiVincenzo's first criterion, that of preparing a pure quantum input state, usually the state of $n$ separate qubits in the computational basis state zero. Instead, in the DQC1 model, only one qubit is prepared in a pure state (computational basis zero state) and the rest of the input qubits exist in the maximally mixed state. This model corresponds to a noisier, more feasible experimental setting and was initially motivated by liquid-state NMR proposals for quantum computing. The DQC1 model was shown to be capable of estimating the coefficients of the Pauli operator expansion efficiently. Following this, Shepherd defined the complexity class `Bounded-error Quantum 1-pure-qubit Polynomial-time (BQ1P)', to capture the power of the DQC1 model~\cite{s06}, and proved that a special case of Pauli operator expansion, the problem of estimating the normalised trace of a unitary matrix to be complete for this class. This problem, and others that can be reduced to it, such as the estimation of the value of the Jones polynomial (see Ref.~\cite{dattashaji11} for more such connections), is interesting from a complexity theoretical point of view since it has no known efficient classical algorithm. Moreover they are not known to belong to the class NP, therefore the problem of verifying the correctness of the result is non-trivial. More recently, it was shown that an ability to simulate classically efficiently a slightly modified version of this model would lead to the collapse of the polynomial hierarchy to the third level \cite{morimae2013hardness}.

The approach of the Verifiable Universal Blind Quantum Computing (VUBQC) \cite{fitzsimons2012unconditionally} is based on the intermediate step of blind computing, a cryptographic protocol where a restricted client runs the desired computation on a powerful server, such that the server does not learn anything about the delegated computation. A protocol for universal blind quantum computation with a client able to prepare only single qubits, based on Measurement-based Quantum Computing (MBQC) \cite{RB01} model was introduced in~\cite{BFK09}. Here, we take the same approach towards verification by first adapting this existing protocol for blind computing to the DQC1 model. Thus, the first goal is to define what it means to have a DQC1 computation in the MBQC setting. Fixing the input state to almost maximally mixed as it is done in the circuit picture of the DQC1 model does not suffice since the required auxiliary qubits for MBQC could potentially increase the number of pure qubits in the system by more than a logarithmic amount~\footnote{Increasing the number of pure qubits in the input to the order of logarithmic in the size of the computation is shown not to add extra power to the one pure qubit complexity class \cite{s06}.}. This adaptation is necessary as currently all the optimal schemes \cite{Dorit,BFK09,DKL11,MDK11,BKBFZW11,FM12,Morimae12,MF12b,SKM13,MPF13,GMMR13} for the blind computation exploit the possibility of adaptive computation based on the measurement, a freedom not allowed in the original DQC1 model. The main results presented in this paper are the following.

\begin{itemize}
\item We introduce a new definition of DQC1 computation within the MBQC framework, called the DQC1-MBQC model~\footnote{We use a different acronym than DQC1 to emphasis the structural distinction with the standard DQC1 model.}, which captures the essential property of its original definition in the circuit model. Moreover, we show that the original definition of complexity class BQ1P is contained in DQC1-MBQC, where the latter is able to capture the process where new qubits are introduced or traced out during the execution of the computation.
\item  We provide a sufficient condition for a graph state (underlying resource for an MBQC computation \cite{HEB04}) to be usable within DQC1-MBQC. A direct consequence of this is that the universal blind protocol, which satisfies this condition, can be directly adapted to the setting where the server is a DQC1-MBQC machine and the client is able to send one single qubits at a time.
\item Building on the blind protocol and adapting the methods presented in \cite{fitzsimons2012unconditionally}, a verification protocol for the class DQC1-MBQC using a verifier restricted to DQC1-MBQC is given, where the probability of the client being forced to accept an incorrect result can be adjusted by setting the security parameter of the model. Since the protocol of \cite{fitzsimons2012unconditionally} does not satisfy the sufficient condition and hence not runnable in the DQC1-MBQC, an alternative method is presented which also leads to different complexity results.
\end{itemize}

\subsection{Preliminaries}

We first introduce the notation necessary to describe a computation in MBQC \cite{RB01,Mcal07}. A generic pattern, consists of a sequence of commands acting on qubits:
\begin{itemize}
\item $N_i(\ket{q})$: Prepare the single auxiliary qubit~$i$ in the state $\ket{q}$;

\item $E_{i,j}$: Apply entangling operator controlled-$Z$ to qubits $i$ and $j$;

\item $M^{\alpha}_i$: Measure qubit $i$ in the basis $\{ \frac{1}{\sqrt{2}} ( \ket{0} + e^{i \alpha} \ket{1} ), \frac{1}{\sqrt{2}} ( \ket{0} - e^{i \alpha} \ket{1} ) \}$ followed by trace out the measured qubit. The result of measurement of qubit $i$ is called signal and is denoted by $s_i$;

\item $X_i^{s_j}, Z_i^{s_j}$: Apply a Pauli $X$ or $Z$ correction on qubit $i$ depending on the result $s_j$ of the measurement on the $j$-th qubit.

\end{itemize}
The corrections could be combined with measurements to perform `adaptive measurements' denoted as $\prescript{s_z}{}{\left[ M^{\alpha}_i \right]}^{s_x} = M^{(-1)^{s_x}\alpha+s_z \pi}_i $. A pattern is formally defined by the choice of a finite set $V$ of qubits, two not necessarily disjoint sets $I \subset V$ and $O \subset V$ determining the pattern inputs and outputs, and a finite sequence of commands acting on $V$.

\begin{definition}\label{def:run} \cite{Flow06} A pattern is said to be runnable if
 \begin{enumerate}
 \item[(R0)] no command depends on an outcome not yet measured;
 \item[(R1)] no command (except the preparation) acts on a measured or not yet prepared qubit;
 \item[(R2)] a qubit is measured (prepared) if and only if it is not an output (input).
 \end{enumerate}
 \end{definition}

 The entangling commands $E_{i,j}$ define an undirected graph over $V$ referred to as $(G,I,O)$. Along with the pattern we define a partial order of measurements and a dependency function $d$ which is a partial function from $O^C$ to $\mathcal{P}^{I^C}$, where $\mathcal{P}$ denotes the power set. Then, $j \in d (i)$ if $j$ gets a correction depending on the measurement outcome of $i$. In what follows, we will focus on patterns that realise (strongly) deterministic computation, which means that the pattern implements a unitary on the input up to a global phase. A sufficient condition on the geometry of the graph state to allow unitary computation is given in \cite{Flow06,gflow07} and will be used later in this paper. In what follows, $x \sim y$ denotes that  $x$ is adjacent to $y$ in $G$.

\begin{definition} \label{d-flow} \cite{Flow06} A \emph{flow} $(f,\preceq)$ for a geometry $(G,I,O)$ consists of a map $f:O^c\mapsto I^c$ and a partial order $\preceq$ over $V$ such that for all $x\in O^c$
\begin{itemize}
\item[(F0)]~~$x \sim f(x)$;
\item[(F1)]~~$x \preceq f(x)$;
\item[(F2)]~~for all $y\not = x, y \sim f(x)$:  $x \preceq y$\,.
\end{itemize}
\end{definition}

\subsection{Main results}

\subsubsection{DQC1-MBQC}

We define the class BQ1P formally as introduced by Shepherd~\cite{s06}, and then recast it into the MBQC framework.

\begin{definition}[Bounded-error Quantum 1-pure-qubit Polynomial-time complexity class ] \label{def-BQ1P} \cite{s06} BQ1P is defined using a bounded-error uniform family of quantum circuits -- DQC1. A DQC1 circuit takes as input a classical string $\boldsymbol{x}$, of size $n$, which encodes a fixed choice of unitary operators applied on a standard input state $\ket{0}\bra{0} \otimes I_{w-1} / 2^{w-1}$. The width of the circuit $w$ is polynomially bounded in $n$. Let $Q_n (x)$ be the result of measuring the first qubit of the final state of a DQC1 circuit. A language in BQ1P is defined by the following rule:
\begin{gather*}
\forall a \in L : Pr(Q_n (a)=1) \ge \frac{1}{2}+\frac{1}{2q(n)} \\
\forall a \notin L : Pr(Q_n (a)=1) \le \frac{1}{2} - \frac{1}{2q(n)}
\end{gather*}
for some polynomially bounded q(n).
\end{definition}

The essential physical property of DQC1 that we mean to preserve in DQC1-MBQC is its limited purity. To capture this we introduce the \emph{purity parameter}:
\begin{equation}
\pi(\rho) = \log_2{ (\text{Tr}(\rho^2))} + d,
\end{equation}
where $d$ is the logarithm of the dimension of the state $\rho$. For a DQC1 circuit with $k$ pure qubits, at each state of the computation the value of purity parameter $\pi$ for that state remains constant equal to $k$. In fact, Shepherd showed that the class BQ1P is not extended by increasing the number of pure input qubits logarithmically. Thus, a purity that does not scale too rapidly with the problem size still remains in the same complexity class.

A characterisation of MBQC patterns compatible with the idea of the DQC1 model as introduced above is presented next. Any MBQC pattern is called DQC1-MBQC when there exists a runnable rewriting of this pattern such that after every elementary operation (for any possible branching of the pattern) the purity parameter $\pi$ does not increase over a fixed constant. We assume that the system at the beginning has only the input state and at the end has only the output state.

We define a new complexity class that captures the idea of one pure qubit computation in the MBQC model. This complexity class, that we name DQC1-MBQC, can be based on any universal DQC1-MBQC resource pattern, which is defined analogously to the DQC1 circuits~\cite{s06} as a pattern that can be adapted to execute any DQC1-MBQC pattern of polynomial size. A particular example of such a resource, as we will present later, can be built using the the brickwork state of \cite{BFK09} designed for the purpose of universal blind quantum computing. The input to a universal pattern is the description of a computation as a measurement angle vector and is used to classically control the measurements of the MBQC pattern.  The quantum input of the open graph is always fixed to a mostly maximally mixed state, in correspondence to the DQC1 model.

\begin{definition}
A language in DQC1-MBQC complexity class is defined based on a universal DQC1-MBQC resource pattern $P_{\alpha}$ that takes as input an angle vector $\boldsymbol{\alpha}$ of size $n$ and is applied on the quantum state $\ket{+}\bra{+} \otimes I_{w-1}/2^{w-1}$, $w \in O(n)$. A word $\boldsymbol{\alpha}$ belongs to the language depending of the probabilities of the measurement outcome ($R_n (\boldsymbol{\alpha})$) of the first output qubit of pattern $P_{\alpha}$ which are defined identically to Definition \ref{def-BQ1P}:
\begin{gather*}
\forall a \in L : Pr(R_n (\alpha)=1) \ge \frac{1}{2}+\frac{1}{2r(n)} \\
\forall a \notin L : Pr(R_n (\alpha)=1) \le \frac{1}{2} - \frac{1}{2r(n)}
\end{gather*}
for some polynomially bounded r(n).
\end{definition}

\begin{corollary}
BQ1P $\subseteq$ DQC1-MQBC.
\end{corollary}

\begin{proof}
Any circuit description using a fixed set of gates can be efficiently translated into a measurement pattern applicable on the brickwork state. A specific example of translating each gate from the universal set \{Hadamard, $\pi /8$, c-NOT\} to a `brick' element of the brickwork state is given in~\cite{BFK09}. The quantum input state in the resulting measurement pattern is in the almost-maximally-mixed state, therefore the pattern is a valid DQC1-MBQC pattern.
\end{proof}

\begin{definition} \label{def:OCQ}
An MBQC pattern is a DQC1-MBQC pattern if there is a runnable sequence of commands where for every elementary command $i$ and any measurement outcomes, there exists a fixed constant value $c$ such that the overall quantum state of the system ($\rho_i$ with dimension $d_i$) after the $i$th operation satisfies the following relation
\begin{equation*}
\pi(\rho_i) < \pi(\rho_{in}) + c
\end{equation*}
Where $\rho_{in}$ is the quantum input of the pattern with dimension $d_{in}$, which is fixed to be the product of $c_{in}$ (constant) pure qubits and a maximally mixed state of $d_{in}-c_{in}$ qubits.
\end{definition}

The above definition captures the essence of DQC1 in that it maintains a low purity, high entropy state in MBQC, in contrast to DiVincenzo's first criterion. We derive a sufficient condition (that is also constructive) for the open graph state leading to DQC1-MBQC, capturing the universal blind quantum computing protocol as a special case. However, a general characterisation and further structural link with determinism in MBQC~\cite{Flow06,gflow07,fFlow10} is left as an open question for future work.

\begin{theorem} \label{t-compact}
Any measurement pattern on an open graph state $(G,I,O)$ with flow $(f,\preceq)$ (as defined in Definition \ref{d-flow}) and measurement angles $\boldsymbol{\alpha}$ where either $|I| = |O|$ or the flow function is surjective and all auxiliary preparations are on the $(X-Y)$ plane represents a DQC1-MBQC pattern.
\end{theorem}

The full details and the proof of this theorem is provided in Section \ref{section_2}.

\subsubsection{Blindness}

A direct consequence of Theorem \ref{t-compact} is that the Universal Blind Computing Protocol (UBQC) introduced in~\cite{BFK09} can be easily adapted to fit within the DQC1-MBQC class, since it is based on an MBQC pattern on a graph state with surjective flow.

In the blind cryptographic setting a client (Alice) wants to delegate the execution of an MBQC pattern to a more powerful server (Bob) and hide the information at the same time. The UBQC protocol is based on the separation of the classical and quantum operations when running an MBQC pattern. The client prepares some randomly rotated quantum states and sends them to the server and from this point on the server executes the quantum operations on them (entangling according to the graph and measuring) and the client calculates the measurement angles for the server and corrects the measurement outcomes she receives (to undo the randomness and get the correct result).

To define blindness formally we allow Bob to deviate from the normal execution in any possible way, and this is captured by modelling his behaviour during the protocol by an arbitrary CPTP map. The main requirement for blindness is that for any input and averaged over all possible choices of parameters by Alice, Bob's final state can always be written as a fixed CPTP map applied on his initial state, thus not offering any new knowledge to him. This definition of stand-alone blindness was presented first in \cite{DFPR13} and takes into account the issue of prior knowledge.

\begin{definition} [Blindness] \label{def:blindness} Let $P$ be a protocol for delegated computation: Alice's input is a description of a computation on a quantum input, which she needs to perform with the aid of Bob and return the correct quantum output. Let $\rho_{AB}$ express the joint initial state of Alice and Bob and   $\sigma_{AB}$ their joint final state, when Bob is allowed to do any deviation from the correct operation during the execution of $P$, averaged over all possible choices of random parameters by Alice. The protocol $P$ is blind iff

\begin{gather}
\forall \rho_{AB} \in \mathcal{L} (\mathcal{H}_{AB}), \exists \mathcal{E}:\mathcal{L}(\mathcal{H}_B) \rightarrow \mathcal{L}(\mathcal{H}_B)  \text{, s.t. } \;\;\; \text{Tr}_A (\sigma_{AB}) = \mathcal{E} (\text{Tr}_A(\rho_{AB})) \label{eq11}
\end{gather}

\end{definition}

To adapt the original UBQC protocol into the DQC1-MBQC setting we change the order of the operations so that the client does not send all the qubits to the server at the beginning, but during the execution of the pattern, following a rewriting of the pattern that is consistent with the purity requirement. The details are described in Section \ref{section_2}.

\begin{theorem}

There exists a blind protocol for any DQC1-MBQC computation where the client is restricted to BPP and the ability to prepare single qubits and the server is within DQC1-MBQC.

\end{theorem}

\subsubsection{Verification}

In the verification cryptographic setting a client (Alice) wants to delegate a quantum computation to a more powerful server (Bob) and accept if the result is correct or reject if the result is incorrect (server is behaving dishonestly). The main idea of the original protocol of~\cite{fitzsimons2012unconditionally} is to test Bob's honesty by hiding a trap qubit among the others in the resource state sent to him by Alice. Blindness means that Bob cannot learn the position of the trap, nor its state. During the execution of the pattern Bob is asked to measure this trap qubit and report the result to Alice. If Bob is honest this measurement gives a deterministic result, which can be verified by Alice. Bob being dishonest means that Alice will receive the wrong result with no-zero probability. Depending on that result, Alice accepts or rejects the final output received by Bob.

To define verifiability formally we need to establish an important difference with the original protocol~\cite{fitzsimons2012unconditionally}: In a DQC1-MBQC pattern the quantum input is in a mixed state as opposed to a pure state. Reverting to the original definition that derives from the quantum authentication schemes in~\cite{BCGST02} we need to add an extra reference system R, that is used to purify the mixed input that exists in Alice's private system A. The assumption is that Bob does not learn anything about the reference system (ex. Alice is provided with the quantum input from a third trusted party which also holds the purification). Bob is allowed to choose any possible cheating strategy and our goal is to minimise the probability of Alice accepting the incorrect output of the computation at the end of the protocol.

\begin{definition} \label{def:verifiability}
A protocol for delegated computation is $\epsilon$-verifiable ($0\leq \epsilon < 1$) if for any choice of Bob's strategy $j$, it holds that for any input of Alice:
\begin{equation}
\text{Tr} (\sum_{\nu} p(\nu) P^{\nu} _{\text{incorrect}} B_j (\nu)) \leq \epsilon
\end{equation}
where $B_j(\nu)$ is the state of Alice's system A together with the purification system R at the end of the run of the protocol, for choice of Alice's random parameters $\nu$ and Bob's strategy $j$. If Bob is honest we denote this state by $B_0 (\nu)$. Let $P_{\perp}$ be the projection onto the orthogonal complement of the the correct (purified) quantum output. Then,
\begin{equation}
P^{\nu} _{\text{incorrect}} = P_{\perp}\otimes \ket{\eta^{\nu_c}_t}\bra{\eta^{\nu_c}_t}
\end{equation}
where $ \ket{\eta^{\nu_c}_t}$ is a state that indicates if Alice accept or reject the result (see Section 3).
\end{definition}

A verification protocol should also be correct, which means that in case Bob is honest Alice's state at the end of the run of the protocol is the correct output of the computation and an extra qubit set in the accept state (this property is also referred to as completeness).

In VUBQC, in order to adjust the parameter $\epsilon$ to any arbitrary value between 0 and 1 (a technique called probability amplification), one needs to add an order of polynomial of the input size, many trap qubits within the MBQC pattern. Specifically, adding polynomial traps and incorporating the pattern into a fault tolerant scheme that corrects $d$ errors, gives parameter $\epsilon$ inversely exponential on $d$. As we explain in Section 3, adding a polynomial number of traps, following the same scheme as VUBQC, creates a pattern that is not DQC1-MBQC. Therefore to achieve an amplification of the error probability we need to develop a modified trapping scheme.

In Section \ref{section_3} we give a verification protocol for DQC1-MBQC problems where instead of running the pattern once, $s$ computations of the same size are run in series, one being the actual computation and the others being trap computations. A similar approach is also considered for the restricted setting of the photonic implementation of VUBQC \cite{barz2013experimental} and a verification protocol of the entanglement states \cite{pappa2012multipartite}. In our setting each trap computation contains an isolated trap injected in a random position between the qubits of the pattern. We prove that in this verification protocol the server is within DQC1-MBQC complexity class, while the client is within BPP together with single qubit preparations (as in the original VUBQC). Moreover in this verification protocol we achieve the goal of probability amplification by choosing the appropriate value for parameter $s$.

\begin{theorem} \label{verifiable_main}
There exists a correct $\epsilon$-verifiable protocol where the client is restricted to BPP and the ability to prepare single qubits and the server is within DQC1-MBQC. Using $O(sm)$ qubits and $O(sm)$ time steps, where $m$ is the size of the input computation, we have:
\begin{equation}
\epsilon = \frac{2m}{s}
\end{equation}

\end{theorem}

\section{DQC1-MBQC and Blindness} \label{section_2}

In this section we give a constructive proof of our main theorem for DQC1-MBQC and show how to construct a blind protocol as a consequence. The first step for proving Theorem \ref{t-compact} is the following rewriting scheme for patterns with flow.

\begin{lemma} \label{l-compact} Any measurement pattern on an open graph state $(G,I,O)$ with flow $(f,\preceq)$ (as defined in Definition \ref{d-flow}) and measurement angles $\boldsymbol{a}$ where either $|I| = |O|$ or the flow function is surjective can be rewritten as

\begin{equation} \label{eq-compact}
P_{\boldsymbol{a}} = \prod_{i \in O} X_{i}^{S_{i}^x} Z_{i}^{S_{i}^z}   \prod_{i \in O^c}^{\preceq} \left( \prescript{S_{i}^z}{}{ \left[ M_{i}^{a_{i}} \right]}^{S_{i}^x} \left( \prod_{\substack{ \{k : k \sim i, k \succeq  i\}}} E_{i,k} \right) N_{f(i)}(\ket{+}) \right)
\end{equation}
where $S_{i}^x=s_{f^{-1}(i)}$ for $i \in I^c$, else $S_{i}^x=0$ and $S_{i}^z= \sum_{\{k:k \in I^c, k \sim i,i \neq f^{-1}(k) \}} s_{f^{-1}(k)} \mod{2}$. The above pattern is runnable and implements the following unitary
\begin{equation} \label{eq-unit}
U_{G,I,O,\boldsymbol{a}} = 2^{|O^c|/2} \left( \prod_{i \in O^c} \bra{+_{a_i}}_i \right) E_G N_{I^c}
\end{equation}
where $E_G$ and $N_{I^c}$ represent the global entangling operator and global preparation respectively.
\end{lemma}

\begin{proof}
First we need to prove that $P_{\boldsymbol{a}}$ is runnable (cf. Definition~\ref{def:run}). For condition (R0) we make the following observations: At step $i$, for $i \in I^c$, we need signal $s_{f^{-1}(i)}$ which is generated at step $f^{-1}(i)$, where $f^{-1}(i) \prec i$ from flow condition (F1). We also need signals $s_{f^{-1}(k)}$, for $\{k:k \in I^c, k \sim i,i \neq f^{-1}(k)\}$, which are generated at step $f^{-1}(k)$, where $f^{-1}(k) \prec i$ from flow condition (F2). Thus, condition (R0) is satisfied (see Figure \ref{fig:figure1} for a particular example). For condition (R1) we make the following observations: At step $i$, for $i \in O^c$, the entangling operator and measurement operator act on qubit $i$ which either belongs in the set of inputs $I$ or is created at step $f^{-1}(i)$, where $f^{-1}(i) \prec i$ from flow condition (F1). Entangling operator acts also on qubits $\{k : k \sim i, k \succeq  i\}$. If $k=f(i)$ then qubit $k$ is created at the same step ($i$) by operator $N_{f(i)}$. If $k \neq f(i)$ then qubit $k$ is either an input or it is created at step $f^{-1}(k)$, and we have by flow condition (F2): $i$ is a neighbour of $k$ and $i \neq f^{-1}(k)$, thus $f^{-1}(k) \prec i$ (Figure \ref{fig:figure1}). Final correction operators act on qubits that belong to the set of outputs $O$, which either belong also to the set of inputs $I$ or are created at steps $\{ f^{-1}(i):  i \in O \}$, where $\forall i \in O \setminus I, f^{-1}(i) \prec i$ from flow condition (F1). Moreover they have not yet been measured since $i \notin O^C$. Thus, condition (R1) is satisfied. It is easy to see that condition (R2) is satisfied.

\begin{figure}[h]
\centering
\includegraphics[scale=0.7]{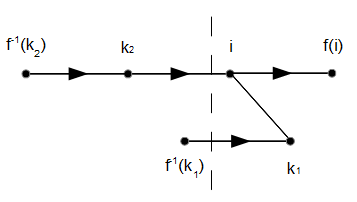}
\caption{Qubit $i$ gets an X correction from $k_2$ and Z corrections from $f^{-1}(k_2)$ and $f^{-1}(k_1)$. Qubits on the left of the dashed line are in the past of $i$. Qubit $k_1$ is created at timestep $f^{-1}(k_1)$ which is before timestep $i$ from flow condition (F2).}
\label{fig:figure1}
\end{figure}

Next we prove that the pattern of Equation \ref{eq-compact} is implementing the unitary operation of Equation \ref{eq-unit} when applied on an open graph with the properties described above. Since condition (R1) is satisfied, all preparation operators trivially commute with all previous operators
\begin{equation*}
P_{\boldsymbol{a}} = \prod_{i \in O} X_{i}^{S_{i}^x} Z_{i}^{S_{i}^z}   \prod_{i \in O^c}^{\preceq} \left( \prescript{S_{i}^z}{}{ \left[ M_{i}^{a_{i}} \right]}^{S_{i}^x} \left( \prod_{\substack{ \{k : k \sim i, k \succeq  i\}}} E_{i,k} \right) \right) N_{I^c}
\end{equation*}
Each entangling operator commutes with all previous measurements since it is applied on qubits with indices $\succeq i$.
\begin{equation*}
P_{\boldsymbol{a}} = \prod_{i \in O} X_{i}^{S_{i}^x} Z_{i}^{S_{i}^z}   \prod_{i \in O^c}^{\preceq} \left( \prescript{S_{i}^z}{}{ \left[ M_{i}^{a_{i}} \right]}^{S_{i}^x}  \right) E_G  N_{I^c}
\end{equation*}
We can decompose the conditional measurements into simple measurements and corrections
\begin{equation*}
P_{\boldsymbol{a}}  =  \prod_{i \in O} X_{i}^{S_{i}^x} Z_{i}^{S_{i}^z}   \prod_{i \in O^c}^{\preceq} \left(  M_{i}^{a_{i}} X_i^{S_{i}^x} Z_i^{S_{i}^z} \right) E_G  N_{I^c}
\end{equation*}
By rearranging the order of correction operators we take
\begin{equation*}
P_{\boldsymbol{a}} = \prod_{i \in O^c}^{\preceq} \left( X_{f(i)}^{s_i} \prod_{\{k:k \sim f(i), k \neq i\}} Z_k^{s_i} M_i^{a_i} \right) E_G N_{I^c}
\end{equation*}
The above equation implements the unitary operation presented in the lemma (Equation \ref{eq-unit}) as proved in \cite{Flow06}.
\end{proof}

Next, we notice that there exist many universal families of open graph states satisfying the requirements of the above lemma. One such example is the brickwork graph state originally defined in \cite{BFK09}. In this graph state (Figure \ref{fig:state1}), the subset of vertices of the first column correspond to the input qubits $I$ and the subset of vertices of the final column correspond to the output qubits $O$. This graph state has flow function $f((i,j))=(i,j+1)$ and the following partial order for measuring the qubits: $\{(1,1),(2,1), \ldots, (w,1)\} \prec \{(1,2),(2,2), \ldots, (w,2)\} \prec \ldots \prec \{(1,d-1),(2,d-1), \ldots, (w,d-1)\}$, where $w$ is the width and $d$ is the depth of the graph and hence from Lemma \ref{l-compact} we obtain the following corollary.

\begin{figure}[h]
\centering
\includegraphics[scale=0.7]{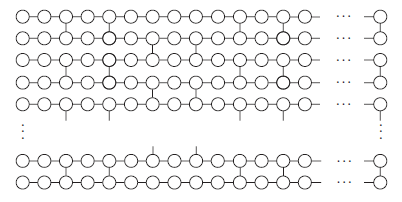}
\caption{Brickwork state}
\label{fig:state1}
\end{figure}

\begin{corollary} \label{c-compact} Any computation over the brickwork open graph state $G$ with qubit index $(i \leq w,j \leq d)$ can be rewritten as follows.
\begin{equation}
P_{\boldsymbol{a}}  =  \prod_{i=1}^{w} X_{(i,d)}^{S_{(i,d)}^x} Z_{(i,d)}^{S_{(i,d)}^z}  \prod_{j=1}^{d-1} \prod_{i=1}^{w}  \prescript{S_{(i,j)}^z}{}{ \left[ M_{(i,j)}^{a_{(i,j)}} \right]}^{S_{(i,j)}^x} \left( \prod_{\substack{ \{k,l : (k,l)\sim (i,j),\\ k \geq i, l \geq j\}}} E_{(i,j),(k,l)} \right) N_{(i,j+1)}  \label{eq8}
\end{equation}

where $S_{(i,j)}^x=s_{(i,j-1)}$ for $j>1$, else $S_{(i,1)}^x=0$

and $S_{(i,j)}^z= \sum_{\{k,l:(k,l)\sim(i,j),l\leq j\}} s_{(k,l-1)} \mod{2}$ for $j>2$, else $S_{(i,j)}^z=0$.

\end{corollary}

We show that patterns defined in Lemma \ref{l-compact} are within the framework of Definition \ref{def:OCQ} hence obtaining a sufficient condition for DQC1-MBQC.

\setcounter{theorem}{0}
\begin{theorem} \label{t-compact}
Any measurement pattern that can be rewritten in the form of Equation \ref{eq-compact} represents a DQC1-MBQC pattern.
\end{theorem}
\setcounter{theorem}{3}

\begin{proof} A first general observation about the purity parameter $\pi$ is that adding a new pure qubit $\sigma$ to state $\rho$ means that $\pi$ increases by unity
\begin{equation*}
\pi_{\rho\otimes\sigma} = \log_2{\text{Tr}((\rho \otimes \sigma) ^2)} + d + 1  = \log_2{\text{Tr}(\rho ^2)\text{Tr}(\sigma^2)} + d + 1  = \pi_{\rho} + 1.
\end{equation*}
Additionally, applying any unitary $U$ does not change the purity parameter $\pi$ of the system since $\text{Tr}((U \rho U^{\dagger})^2) = \text{Tr}(\rho^2)$ and dimension remains the same.

Returning to Equation \ref{eq-compact}, we notice that for every step $i\in O^c$ of the product the total computation performed  corresponds mathematically to the following: On the qubit tagged with position $i$, a $J(a'_{i})$ unitary gate is applied (where $a'_{i}$ is an angle that depends on $a_{i}$ and previous measurement results) up to a specific Pauli correction (depending on the known measurement result) and some specific Pauli corrections on the its entangled neighbours (again depending on the measurement result). At the end the qubit is tagged with position $f(i)$ (where $f$ is the flow function). Since this mathematically equivalent computation is a unitary and the dimension of the system remains the same (there is only a change of position tags) we conclude that each step $i \in O^c$ does not increase the purity parameter of the system. To finish the proof we need to ensure that the individual operations within each step $i \in O^c$ and for $i \in O$ do not increase the purity parameter by more than a constant (and since there is only a constant number of operations within each step this does not increase the purity at any point more than constant). This is true since all these operations apply on (or add or trace over) a constant number of qubits.
\end{proof}

Building on this result, we can translate the UBQC protocol of \cite{BFK09} (and in fact many other existing protocols) to allow the blind execution of any DQC1-MBQC computation, where the server is restricted to DQC1-MBQC complexity class. The UBQC protocol is based on the brickwork graph state described above. Alice prepares all the qubits of the graph state, adding a random rotation around the $(X,Y)$ plane to each one of them: $\ket{+_{\theta_i}}$, where $\theta_i$ is chosen at random from the set $A=\{0,\pi/4,\pi/2, 3\pi/4, \pi, 5\pi/5, 3\pi/2, 7\pi/4\}$ and sends them to Bob, who entangles them according to the graph. The protocol then follows the partial order given by the flow: Alice calculates the corrected measurement angle $\alpha'_i$ for each qubit using previous measurement results according to the flow dependences. She sends to Bob measurement angle $\delta_i=\alpha'_i+\theta_i+r_i \pi$, using an extra random bit $r_i$. Bob measures according to $\delta_i$, reports the result back to Alice who corrects it by XOR-ing with $r_i$. In the case of quantum output, the final layer is sent to Alice and is also corrected according to the flow dependences by applying the corresponding Pauli operators.

\begin{algorithm}
\caption{Blind BQ1P protocol}
\label{prot:blind}
\begin{flushleft}
\textbf{Alice's input:}
\vspace{-7pt}
\end{flushleft}
\begin{itemize}
 \item A vector of angles $\boldsymbol{a}=(a_{1,1}, \ldots, a_{w,d})$, where $a_{i,j}$ comes from the set $A=\{0, \pi/4, 2\pi/4, \ldots, 7\pi/4 \}$, that when plugged in the measurement pattern $P_{\boldsymbol{a}}$ of Equation \ref{eq8} applied on the brickwork state, implements the desired computation. This computation is applied on a fixed input state $\ket{+}\bra{+} \otimes I_{w-1}/2^{w-1}$.
\end{itemize}
\begin{flushleft}
\textbf{Alice's output:}
\vspace{-7pt}
\end{flushleft}
\begin{itemize}
 \setlength{\itemsep}{-1pt}
 \item The top output qubit (qubit in position $(1,d)$).
\end{itemize}
\begin{flushleft}
\textbf{The protocol}
\vspace{-7pt}
\end{flushleft}
\begin{enumerate}
\item Alice picks a random angle $\theta_{1,1} \in A$, prepares one pure qubit in state $R_z(\theta_{1,1}) \ket{+}$ and sends it to Bob who tags it as qubit $(1,1)$.
\item Bob prepares the rest of input state (qubits $(2,1), \ldots, (w,1)$) in the maximally mixed state $I_{w-1}/2^{w-1}$.
\item Alice and Bob execute the rest of the computation in rounds. For $j=1$ to $d-1$ and for $i=1$ to $w$
\begin{enumerate}
 \setlength{\itemsep}{-1pt}
\item  \label{step:client}  \textbf{Alice's preparation}

\begin{enumerate}

\item Alice picks a random angle $\theta_{i,j+1} \in A$.

\item Alice prepares one pure qubit in state $R_z(\theta_{i,j+1}) \ket{+}$.

\item Alice sends it to Bob. Bob tags it as qubit $(i,j+1)$.

\end{enumerate}

\item  \label{step:server} \textbf{Entanglement and measurement}

\begin{enumerate}
\setlength{\itemsep}{-1pt}
\item Bob performs the entangling operator(s): $$\prod_{\substack{ \{k,l : (k,l)\sim (i,j) ,k \geq i, l \geq j\}}} E_{(i,j),(k,l)}$$

\item Bob performs the rest of the computation using classical help from Alice:

\begin{enumerate}
\setlength{\itemsep}{-1pt}
\item Alice computes the corrected measurement angle $a ' _{i,j} = (-1)^{S_{i,j}^x}a_{i,j}+S_{i,j}^z \pi$.

\item Alice chooses a random bit $r_{i,j}$ and computes $\delta _{i,j}=a ' _{i,j} + \theta _{i,j} + r_{i,j} \pi$.

\item Alice transmits $\delta _{i,j}$ to Bob.

\item Bob performs operation $M_{i,j}^{\delta _{i,j}}$ which measures and traces over the qubit $(i,j)$ and retrieves result $b_{i,j}$.

\item Bob transmits $b_{i,j}$ to Alice.

\item Alice updates the result to $s_{i,j}= b_{i,j}+ r_{i,j} \mod{2}$.

\end{enumerate}

\end{enumerate}

\end{enumerate}
\item Bob sends to Alice the final layer of qubits, Alice performs the required corrections and outputs the result.

\end{enumerate}

\end{algorithm}

Since the brickwork graph state satisfies the requirements of Theorem \ref{t-compact} we can adapt the Universal Blind Quantum Computing protocol by making Alice and Bob follow the order of Equation \ref{eq8} and operate on input $\ket{+}\bra{+} \otimes I_{w-1}/2^{w-1}$. A detailed description is given in Protocol \ref{prot:blind}.

\begin{theorem}

Protocol \ref{prot:blind} is correct.

\end{theorem}

\begin{proof}
Correctness comes from the fact that what Alice and Bob jointly compute is mathematically equivalent to performing the pattern of Equation \ref{eq8} on input $\ket{+}\bra{+} \otimes I_{w-1}/2^{w-1}$. The argument is the same as in the original universal blind quantum computing protocol \cite{BFK09} repeated here for completeness. Firstly, since entangling operators commute with $R_z$ operators, preparing the pure qubits in a rotated state does not change the underlying graph state; only the phase of each qubit is locally changed, and it is as if Bob had performed the $R_z$ rotation after the entanglement. Secondly, since a measurement in the $\ket{+_{a}}, \ket{-_{a}}$ basis on a state $\ket{\phi}$ is the same as a measurement in the $\ket{+_{a+\theta}}, \ket{-_{a+\theta}}$ basis on $R_z(\theta)\ket{\phi}$, and since $\delta=a ' + \theta + \pi r$ , if $r = 0$, Bob's measurement has the same effect as Alice's target measurement; if $r = 1$, all Alice needs to do is flip the outcome.

\end{proof}

Note that Protocol  \ref{prot:blind} can be trivially simplified by omitting all the measurements that are applied on maximally mixed states (i.e. all measurements applied on qubits in rows 2 to $w$ from the beginning of the computation until each one is entangled with a non-maximally mixed qubit). However this does not give any substantial improvement in the complexity of the protocol.

\begin{theorem} \label{thm:blindness}

Protocol  \ref{prot:blind} is blind.

\end{theorem}

\begin{proof} [(Proof Sketch)]

A detailed proof is provided in Appendix \ref{appendix_1}. Intuitively, rotation by angle $\theta_{i,j}$ serves the purpose of hiding the actual measurement angle, while rotation by $r_{i,j} \pi$ hides the result of measuring the quantum state. This proof is consistent with definition of blindness based on the relation of Bob's system to Alice's system which takes into account prior knowledge of the secret and is a good indicator that blindness can be composable \cite{DFPR13}.

\end{proof}

Regarding the complexity of the protocol, Alice needs to pick a polynomially large number of random bits and perform polynomially large number of modulo additions that is to say Alice classical computation is restricted to the class $BPP$. However Alice's quantum requirement is only to prepare single qubits, she has access to no quantum memory or quantum operation. Therefore assuming $BQ1P \not \subset BPP$ suggests  Alice's quantum power is more restricted than $BQ1P$ and hence DQC1-MBQC. On the other hand, Bob performs a pattern of the form given in Equation \ref{eq8}, with the difference that instead of preparing the pure qubits himself, he receives the pure qubits through the quantum channel that connects him with Alice. Also, the qubits are not prepared in state $\ket{+}$, but in some state on the $(X,Y)$ plane, but this doesn't alter the reasoning in the complexity proofs. Thus, Bob has computational power that is within the DQC1-MBQC complexity class according to the Corollary \ref{c-compact} and Theorem \ref{t-compact}.

\section{Verification} \label{section_3}

VBQC protocol is based on the ability to hide a trap qubit inside the graph state while not affecting the correct execution of the pattern. Both the trap qubit and the qubits which participate in the actual computation are prepared in the $(X,Y)$ plane of the Bloch sphere. To keep them disentangled, some qubits (called dummy) prepared in the computational basis $\{\ket{0},\ket{1}\}$, are injected between them. Being able to choose between the two states is essential for blindness (Theorem 4 in~\cite{fitzsimons2012unconditionally}). In particular, if a dummy qubit is in state $\ket{0}$, applying the entangling operator $cZ$ between this qubit and a qubit prepared on the $(X,Y)$ plane has no effect. If a dummy qubit is in state $\ket{1}$  then applying $cZ$ will introduce a Pauli $Z$ rotation on the qubit prepared on the $(X,Y)$ plane. This effect can be cancelled by Alice in advance, by introducing a Pauli $Z$ rotation on all the neighbours of $\ket{1}$'s when preparing the initial state.

In the simplest version of VUBQC, a single trap, prepared in state $\ket{+_{\theta_t}}$, where $\theta$ is chosen at random from the angles set $A$ (defined above) and placed at position $t$, chosen at random between all the vertices of the open graph state $(G,I,O)$. During the execution of the pattern, if $t \notin O$,  Bob is asked to measure qubit $t$ with angle $\theta_t+r \pi$ and return the classical result $b_t$ to Alice. If $b_t=r_t$ Alice sets an indicator bit to state $acc$ (which means that this computation is accepted), otherwise she sets it to $rej$ (computation is rejected). If $t \in O$, Alice herself measures the trap qubit and sets the indicator qubit accordingly. This version of the protocol is proven to be correct and $\epsilon$-verifiable, with $\epsilon=(m-1)/m$, where $m$ is the size of the computation.

A generalisation of this technique which allows for arbitrary selection of parameter $\epsilon$ is also presented in~\cite{fitzsimons2012unconditionally}. By allowing for a polynomial number of traps to be injected in the graph state and adapting the computation inside a fault tolerant scheme with parameter $d$ one can have $\epsilon$ inversely exponential to $d$. The question is whether this amplification method can also be used to design a verification protocol for DQC1-MBQC with arbitrary small $\epsilon$. Unfortunately the underlying graph state used by this protocol does not have flow and not all qubits are prepared in the $(X,Y)$ plane, so that one can not apply Theorem \ref{t-compact} to get a compatible rewriting of the pattern. Moreover, having the requirement that we should be able to place every trap qubit (which is a pure qubit) at any position in the graph, means that there exist patterns that will never be possible to be rewritten to satisfy the purity requirement. This leads us to seek a different approach for probability amplification for verification in the DQC1-MBQC model.

\begin{figure}[t]
\centering
\includegraphics[scale=0.5]{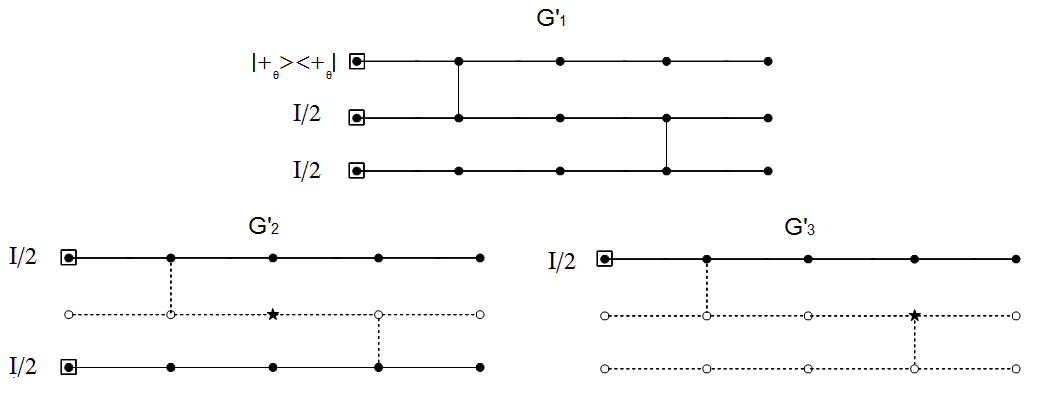}
\caption{Let $G'$ be the graph which consists of $s$ isolated brickwork graphs (each denoted as $G'_i$), each of the same dimensions required for the desired computation. An example construction with $s=3$ and one trap per graph together with a small brickwork state for computation is given above. Black vertices correspond to auxiliary qubits prepared on the $(X-Y)$ plane or mixed state when they are inputs (inside square), star vertices correspond to trap qubits and white vertices to auxiliary qubits prepared in the computational basis. Edges represent entangling operators, dashed where entangling has no effect (except of local rotations).}
\label{fig:figure2}
\end{figure}

Instead of placing a polynomial number of isolated traps within the same graph, which is also used to perform the actual computation, we utilise $s$ isolated brickwork subgraphs, one used for the computation and the rest being trap subgraphs (see Figure \ref{fig:figure2}). Therefore at the beginning of the protocol, Alice chooses random parameter $t_g$, which denotes which graph will be the computational subgraph, and for each of the remaining trap subgraphs $i$, she chooses a random position $t_i$ to hide one isolated trap. The rest of each trap subgraph will be a trivial computation (all measurement angles set to 0) on a totally mixed state, and a selected set of dummy qubits are placed to isolate this computation from the trap. Computation subgraph and trap subgraphs are of the same size, and by taking advantage of the blindness of the protocol, Bob cannot distinguish between them. Therefore, to be able to cheat, he needs to deviate from the correct operation only during the execution on the computational subgraph and never deviate while operating on any of the traps. This gives the desirable $\epsilon$ parameter that will be proved later. The full description of protocol is given in Protocol \ref{prot:verifiable_many}. Each isolated pattern $k$ is executed separately and according to the DCQ1-MBQC rewriting on the brickwork state given in Equation~\ref{c-compact} in the blind setting. Pre-rotations on the neighbours of dummy qubits guarantee that the computation is not affected by the choice of dummies as described before. To prove the complexity of the protocol we need to notice that although the graph used satisfies the conditions of Theorem \ref{t-compact}, the existence of the dummy qubits prepared in the computational basis creates the need of a new proof.

\begin{algorithm}
\caption{Verifiable DQC1-MBQC protocol with $s-1$ trap computations}
\label{prot:verifiable_many}
\begin{flushleft}
\textbf{Alice's input:}
\vspace{-7pt}
\end{flushleft}
\begin{itemize}
 \item An angle vector $\boldsymbol{a}=(a_{1,1}, \ldots, a_{w,d-1})$, where $a_{i,j}$ comes from the set $A=\{0, \pi/4, 2\pi/4, \ldots, 7\pi/4 \}$, that, when plugged in the measurement pattern $P_{\boldsymbol{a}}$ of Equation \ref{eq8} on the brickwork open graph state $G$ of dimension $(w,d)$ and flow ($f,\preceq$), it implements the desired computation on fixed input $\ket{+}\bra{+} \otimes I_{w-1}/2^{w-1}$.
\end{itemize}
\begin{flushleft}
\textbf{Alice's output:}
\vspace{-7pt}
\end{flushleft}
\begin{itemize}
 \setlength{\itemsep}{-1pt}
\item The top output qubit of $G$ (qubit in position $(1,d)$ in $G$) together with a 1-bit, named \emph{acc}, that indicates if the result is accepted or not.
\end{itemize}
\begin{flushleft}
\textbf{The protocol}
\vspace{-7pt}
\end{flushleft}
\begin{itemize}
\item \textbf{Preparation steps. } Alice picks $t_g$ at random from $\{1, \ldots, s\}$. Let $G'$ be the graph which consists of $s$ isolated brickwork graphs, each of the dimension the same as $G$. Then the $t_g$-th isolated graph (named $G'_{t_g}$) will be the computational subgraph for this run of the protocol.
\item Alice maps the measurement angles of the computational subgraph $G'_{t_g}$ to angles of graph $G$: $\boldsymbol{a}'_{G_{t_g}\setminus O_{t_g}}=\boldsymbol{a}$ and appropriately set the dependency sets $S^x$ and $S^z$ for all the vertices of $G'_{t_g}$ (according to the standard flow), while for the rest of the vertices (graph $G' \setminus G'_{t_g}$) the sets $S^x$ and $S^z$ are empty.
\item For $k=1$ to $s$ except $t_g$:
\begin{enumerate}
\item Alice chooses one random vertex $\boldsymbol{t}_k=(t_x,t_y)_k$ among all vertices of $G'_k$ for placing the trap.\label{step_2}
\item By $G'_k$'s geometry, vertex $(t_x,t_y)$ may be connected by a vertical edge to vertex $(t'_x,t_y)$, where $t'_x$ represents either $t_x+1$ or $t_x-1$. We add in $D$ (set of dummies) all vertices of rows $t_x$, $t'_x$ (if it exists) of $G'_k$, except the trap itself. \label{step_3}
\item All elements of $\boldsymbol{a}'_{G_k}$ are mapped to 0.\label{step_4}
\end{enumerate}
\item Alice chooses random variables $\boldsymbol{\theta}_{G' \setminus D}$, each uniformly at random from $A$.
\item Alice chooses random variables $\boldsymbol{r}_{G'}$ and $\boldsymbol{d}_{D}$, each uniformly at random from $\{0,1\}$.
\item For $k=1$ to $s$:
\begin{enumerate}
\item \textbf{Initial step.} If $k = t_g$ then: Let $(1,1)_k$ be the position of the top input qubit in $G'_k$. Alice prepares the following states and sends them to Bob:
\begin{eqnarray*}
\{(1,1)_k\} & ~~  \ket{+_{\theta _{(1,1)_k}}}\\
\forall (i,1)_k \notin \{(1,1)_k\} & ~~ I/2
\end{eqnarray*}
Otherwise: Alice prepares the following states and sends them to Bob:
\begin{eqnarray*}
\forall (i,1)_k \in D & ~~ \ket{d_{(i,1)_k}} \\
 (i,1)_k = \boldsymbol{t}_k & ~~ \prod_{\{m,l:(m,l)_k \sim (i,1)_k, (m,l)_k \in D\}} Z^{d_{(m,l)_k}} \ket{+_{\theta _{(i,1)_k}}}\\
\forall (i,1)_k \notin \{ D, \boldsymbol{t}_k \} & ~~ I/2
\end{eqnarray*}
\end{enumerate}
\end{itemize}

\end{algorithm}

\addtocounter{algorithm}{-1}

\begin{algorithm}
\caption{\textbf{(cont'd)}}
\label{prot:verifiable2}

\begin{itemize}
\item[]
\begin{enumerate}
\setcounter{enumi}{1}
\item \textbf{Main Iteration.} For $j=1$ to $d-1$, for $i=1$ to $w$:
\begin{enumerate}
 \setlength{\itemsep}{-1pt}
\item  \label{step:client}  \textbf{Alice's preparation}

\begin{enumerate}

\item Alice prepares one pure qubit in one of the following states, depending on $(i,j+1)_k$:
\begin{eqnarray*}
(i,j+1)_k \in D & ~~ \ket{d_{(i,j+1)_k}} \\
(i,j+1)_k \notin D & ~~ \prod_{\{m,l : (m,l)_k \sim (i,j+1)_k, (m,l)_k \in D \}} Z^{d_{(m,l)_k}} \ket{+_{\theta _{(i,j+1)_k}}}
\end{eqnarray*}

\item Alice sends it to Bob. Bob labels it as qubit $(i,j+1)_k$.

\end{enumerate}

\item  \label{step:server} \textbf{Entanglement and measurement}

\begin{enumerate}
\setlength{\itemsep}{-1pt}
\item Bob performs the entangling operator(s): $$\prod_{\substack{ \{m,l : (m,l)_k \sim (i,j)_k ,m \geq i, l \geq j\}}} E_{(i,j)_k,(m,l)_k}$$

\item Bob performs the rest of the computation using classical help from Alice:

\begin{enumerate}
\setlength{\itemsep}{-1pt}
\item Alice computes the corrected measurement angle $a '' _{(i,j)_k} = (-1)^{S_{(i,j)_k}^x}a'_{(i,j)_k}+S_{(i,j)_k}^z \pi$.

\item Alice computes actual measurement angle $\delta _{(i,j)_k}=a '' _{(i,j)_k} + \theta _{(i,j)_k} + r_{(i,j)_k} \pi$.

\item Alice transmits $\delta _{(i,j)_k}$ to Bob.

\item Bob performs operation $M_{(i,j)_k}^{\delta _{(i,j)_k}}$ which measures and traces over the qubit $(i,j)_k$ and retrieves result $b_{(i,j)_k}$.

\item Bob transmits $b_{(i,j)_k}$ to Alice.

\item Alice updates the result to $s_{(i,j)_k}= b_{(i,j)_k}+ r_{(i,j)_k} \mod{2}$.

\end{enumerate}

\end{enumerate}

\end{enumerate}
\item Bob sends the final layer to Alice and Alice applies the final corrections if needed (only in round $t_g$).
\item If the trap qubit is within the qubits received, Alice measures it with angle $\delta_{\boldsymbol{t}_k} = \theta _{\boldsymbol{t}_k} + r_{\boldsymbol{t}_k} \pi$ to obtain $b_{\boldsymbol{t}_k}$. Also, Alice discards all qubits received by Bob in this round except qubit $({1,d})_{t_g}$.
\end{enumerate}
\item Alice outputs qubit in position $({1,d})_{t_g}$ and sets bit \emph{acc} to 1 if $b_{\boldsymbol{t}_k}=r_{\boldsymbol{t}_k}$ for all $k$.
\end{itemize}

\end{algorithm}

\begin{theorem}
The computational power of Bob in Protocol \ref{prot:verifiable_many} is within DQC1-MBQC.
\end{theorem}

\begin{proof}
Note that the $s$ patterns are executed in series and Bob does not keep any qubits between executions. The inputs to these patterns are almost maximally mixed, in accordance with the purity requirement and this `mixedness' propagates through both computational and trap subgraphs.  For the computational subgraph (which is not entangled with the rest) the reasoning of the proof of Theorem \ref{t-compact} applies, since this subgraph satisfies the sufficient conditions and no dummy qubits are used. In the case of a trap subgraph $k$ consider first those operations that apply on the isolated trap and dummy subgraph only. Then for each step $(i,j)_k$ of the main iteration of the protocol (where $(i,j)_k$ is a trap or a dummy) a new pure qubit is sent to Bob, which increases the purity parameter by 1. Entangling will not have any effect on the purity parameter. While the measurement does not increase the purity of the qubit since it was already pure (dummy or trap remain always pure through the computation), and tracing out the resulting qubit will decrease the purity by 1. Thus, the whole step will not change the purity. On the other hand, for the remaining operations the reasoning of the proof of Theorem \ref{t-compact} goes through, since this subgraph satisfies the sufficient conditions. Also operations that apply on both subgraphs are all unitaries therefore they do not affect purity.
\end{proof}

Using the definition of verifiability given in Definition \ref{def:verifiability} we prove the main theorem for the existence of a correct and verifiable DQC1-MBQC protocol (Theorem \ref{verifiable_main}). The full proof is given in Appendix \ref{appendix_2}, while here we describe the main steps.

\begin{proof}[Proof of Theorem \ref{verifiable_main} (Sketch)]

Correctness of Protocol \ref{prot:verifiable_many} comes from the fact that the computational subgraph is disentangled from the rest of the computation and if Bob performs the predefined operations, from the correctness of the blind protocol Alice will receive the correct output. Also, in this case, (and since the traps are corrected to cancel the effect of their entanglement with their neighbouring dummies) the measurement of the traps will give the expected result and Alice will accept the computation.

The proof of verifiability follows the same general methodology of the proof of the original VUBQC protocol \cite{fitzsimons2012unconditionally}, except the last part which contains the counting arguments. For the rest we use single indexing for the qubits, where subgraph $G'_i$ consists of $m$ qubits indexed $(i-1)+1$ to $im$. Therefore the total number of qubits in the protocol is $sm$. Parameter $n$ represents the size of the input of each subgraph (parameter $w$ in the protocol).


Based on Definition \ref{def:verifiability} we need to bound the probability of the (purified) output collapsing onto the wrong  subspace and accepting that result. To explicitly write the final state $B_j (\nu)$ we need to define the following notations. Alice's chosen random parameters are denoted collectively by $\nu$, a subset of those are related to the traps: $\nu_T$ including $t_g$, $t_k$'s and $\theta_{t_k}$'s for $k \in \{1, \ldots, s\}\setminus t_g$. Also $\nu_C = \{\nu \setminus \nu_T\}$. The projection onto the correct state for each trap $t_k$ is denoted by $\ket{\eta_{t_k}^{\nu_T}}$, where $\ket{\eta_{t_k}^{\nu_T}} = \ket{+_{\theta_{t_k}}}$ when $t_k \in O_k$ and $\ket{\eta_{t_k}^{\nu_T}} = \ket{r_{t_k}}$ otherwise (since the trap has been already measured). $C_{\boldsymbol{r}}$ denotes the Pauli operators that map the output state of the computational subgraph to the correct one.  $c_{\boldsymbol{r}}$ is used to compactly deal with the fact that in the protocol each measured qubit $i$ is decrypted by XOR-ing them with $r_i$, except for the trap qubits which remain uncorrected: $\forall k: (c_{\boldsymbol{r}})_{t_k}=0$. $\rho_{M^{\nu}_k}$ denotes the density matrix representing the total quantum state received by Bob from Alice for each round $k$ of the protocol. A special case is the $t_k$th round  where $\rho_{M^{\nu}_k}$ represents the total state received by Bob together with its purification (not known to Bob). The classical information received by Bob at each elementary step $i$ (measurement angles) are represented by $\ket{\delta_i}$'s.

We allow Bob to have an arbitrary deviation strategy $j$, at each elementary step $i$ which is represented as CPTP map $\mathcal{E}_{i}^j$,  followed  by a Pauli $Z$ measurement of qubit $i$ (since Bob has to produce a classical bit at each step and return it to Alice), which is represented by taking the sum over projectors on the computational basis $\ket{b_i}$, for $b_i \in \{0,1\}$. All measurement operators can be commuted to the end of the computation and all CPTP maps can be gathered to a single map $\mathcal{E}^j$ after Bob has received everything from Alice, so that the failure probability can be written as:
\begin{multline}
p_{incorrect}  =   \sum_{\boldsymbol{b'},\nu} p(\nu) \text{Tr} (P_{\perp}  \bigotimes_{k=1}^{s} \ket{\eta_{t_k}^{\nu_T}} \bra{\eta_{t_k}^{\nu_T}} \nonumber \\
C^{\boldsymbol{b'},\nu_C} \ket{\boldsymbol{b'} + \boldsymbol{c}^{\boldsymbol{r}}}\bra{\boldsymbol{b'}}\mathcal{E}^{j} \left( \bigotimes_{k=1}^{s} \bigotimes_{i=1}^{m-n} \ket{\delta_{(k-1)m+i}^{\boldsymbol{b'},\nu}}\bra{\delta_{(k-1)m+i}^{\boldsymbol{b'},\nu}}  \otimes \rho_{M^{\nu}_k}  \right)   \ket{\boldsymbol{b'}}\bra{\boldsymbol{b'} + \boldsymbol{c}^{\boldsymbol{r}}} C^{\boldsymbol{b'},\nu_C \dagger}) \nonumber
\end{multline}
Our strategy will be to rewrite this probability by introducing the correct execution of the protocol before the attack, on each subgraph $k$: $\mathcal{P}_k = \bigotimes_{i=1}^{m-n} (H_{(k-1)m+i} Z_{(k-1)m+i}(\delta_{(k-1)m+i})) E_{G'_k} $ and at the same time decomposing the attack to the Pauli basis, using general Paulis $\sigma_{i,k}$ applying on qubits $(k-1)m+1 \leq \gamma \leq km$ for each $k$.
\begin{multline}
p_{incorrect}  =   \sum_{\boldsymbol{b'},\nu, v , i, j} \alpha_{vi} \alpha_{vj}^* p(\nu) \text{Tr} (P_{\perp}  \bigotimes_{k=1}^{s} \ket{\eta_{t_k}^{\nu_T}} \bra{\eta_{t_k}^{\nu_T}}
C^{\boldsymbol{b'},\nu_C} \ket{\boldsymbol{b'} + \boldsymbol{c}^{\boldsymbol{r}}}\bra{\boldsymbol{b'}} \nonumber \\ \bigotimes_{k=1}^{s} ( \sigma_{i,k} \left( \mathcal{P}_k  \bigotimes_{i=1}^{m-n} \ket{\delta_{(k-1)m+i}^{\boldsymbol{b'},\nu}}\bra{\delta_{(k-1)m+i}^{\boldsymbol{b'},\nu}}  \otimes \rho_{M^{\nu}_k}  \mathcal{P}^{\dagger}_k \right) \sigma_{j,k}) \ket{\boldsymbol{b'}}\bra{\boldsymbol{b'} + \boldsymbol{c}^{\boldsymbol{r}}} C^{\boldsymbol{b'},\nu_C \dagger} \nonumber
\end{multline}
This way we can characterise which Pauli attacks give non-zero failure probability when the final state is projected on the correct one. For convenience we introduce the following sets for an arbitrary Pauli $\sigma_{i,k}$:

\begin{eqnarray}
A_{i,k} = \{ \gamma \text{ s.t. } \sigma_{i|\gamma}=I \text{ and } (k-1)m+1 \leq \gamma \leq km \} \nonumber \\
B_{i,k} = \{ \gamma \text{ s.t. } \sigma_{i|\gamma}=X \text{ and } (k-1)m+1 \leq \gamma \leq km \} \nonumber \\
C_{i,k} = \{ \gamma \text{ s.t. } \sigma_{i|\gamma}=Y \text{ and } (k-1)m+1 \leq \gamma \leq km \} \nonumber \\
D_{i,k} = \{ \gamma \text{ s.t. } \sigma_{i|\gamma}=Z \text{ and } (k-1)m+1 \leq \gamma \leq km \} \nonumber
\end{eqnarray}
We use the superscript $O$ to denote subsets subject to the constraint $km \geq \gamma \geq km-n+1$. For an arbitrary $t_g$, the only attacks that give the corresponding term of the sum not equal to zero: are those that (i) produce an incorrect measurement result for qubits $(t_g-1) m +1 \leq \gamma \leq t_gm-n$ or (ii) operate non-trivially on qubits $ t_g m - n < \gamma \leq t_g m$.
We denote this condition by $i\in E_{i,t_g}$ and $j\in E_{j,t_g}$:  $|B_{i,t_g}|+|C_{i,t_g}|+|D_{i,t_g}^O| \geq 1$ and $|B_{j,t_g}|+|C_{j,t_g}|+|D_{j,t_g}^O| \geq 1$.

The next step will be to characterise which attacks of these subsets remain undetected by the trap mechanism and try to find an upper bound on their contribution to the failure probability. By applying blindness and observing that only the terms where $\sigma_{i,k}=\sigma_{j,k}$ contribute we obtain the following upper bound (details in Appendix \ref{appendix_2}):

\begin{multline}
p_{incorrect}  \leq \sum_{t_g} \sum_{v, i\in E_{i,t_g}} | \alpha_{vi} |^2 p(t_g) \prod_{k=\{1, \ldots, s\} \setminus t_g}  ( \sum_{\substack{km-n < t_k \leq km,\\ \theta_{t_k}}} p(t_k, \theta_{t_k}) ( \bra{+_{\theta_{t_k}}}  \sigma_{i|t_k} \ket{+_{\theta_{t_k}}}   )^2 \nonumber \\
+ \sum_{\substack{(k-1)m < t_k \leq km-n,\\  r_{t_k}}}  p(t_k,  r_{t_k}) ( \bra{r_{t_k}}  \sigma_{i|t_k}\ket{r_{t_k}} )^2 ) \nonumber
\end{multline}

The rest is based on a counting argument using $\forall k$, $|A_{i,k}|+|B_{i,k}|+|C_{i,k}|+|D_{i,k}|=m$.

\begin{gather}
\begin{split}
p_{incorrect}  \leq \sum_{t_g} \sum_{v, i\in E_{i,t_g}} | \alpha_{vi} |^2 \frac{1}{s} \prod_{k=\{1, \ldots, s\} \setminus t_g} \frac{1}{2m}(  2 |A_{i,k}| + |B_{i,k}^O| +|C_{i,k}^O|+ 2 |D_{i,k} \setminus D_{i,k}^O |)  \nonumber \\
\leq  \sum_{t_g} \sum_{v, i\in E_{i,t_g}} | \alpha_{vi} |^2 \frac{1}{s} \prod_{k=\{1, \ldots, s\} \setminus t_g} \frac{1}{2m}(  2m - |B_{i,k}| - |C_{i,k}| - |D_{i,k}^O|) \nonumber
\end{split}
\end{gather}

We denote the product term $\prod_{k=\{1, 2, 3, \ldots, s\}\setminus z} \frac{1}{2m}(  2m - |B_{i,k}| - |C_{i,k}| - |D_{i,k}^O|)$ as $P_{i,z}$. We also denote each set $\{E_{i,1}^* \cap E_{i,2}^* \cap \ldots \cap E_{i,s}^*\}$, where each term $E_{i,w}^*$ is either $E_{i,w}$ or its complement, $E_{i,w}^C$, depending on whether the $w$-th value of a binary vector $\boldsymbol{y}$ (size $s$) is 1 or 0 respectively, as $W_{i,\boldsymbol{y}}$. Let the function $\# \boldsymbol{y}$ give the number of positions $i$ such that $y_i$=1.

\begin{equation}
= \frac{1}{s} ( \sum_{k=1}^s \sum_{\{\boldsymbol{y} : \#\boldsymbol{y}=k\}} \sum_{i \in W_{i,\boldsymbol{y}},v} ( | \alpha_{vi} |^2 \sum_{\{z:y_z=1\}} P_{i,z})) \nonumber
\end{equation}

The condition $ i \in W_{i,\boldsymbol{y}}$ means that the following conditions hold together: $\{ |B_{i,w}|+|C_{i,w}|+|D_{i,w}^O| \geq 1 : y_w=1\}$,$\{ |B_{i,w}|+|C_{i,w}|+|D_{i,w}^O| = 0 : y_w=0\}$.

\begin{align}
\leq  \frac{1}{s} ( \sum_{k=1}^s \sum_{\{\boldsymbol{y} : \#\boldsymbol{y}=k\}} \sum_{i \in W_{i,\boldsymbol{y}},v}  | \alpha_{vi} |^2 k \left(\frac{2m-1}{2m}\right)^{k-1}) = \frac{1}{s} ( \sum_{k=1}^s c_k k \left(\frac{2m-1}{2m}\right)^{k-1} ) \nonumber
\end{align}

where $c_k = \sum_{\{\boldsymbol{y} : \#\boldsymbol{y}=k\}} \sum_{i \in W_{i,\boldsymbol{y}},v}  | \alpha_{vi} |^2$.

An upper bound on the above expression is:

\begin{equation}
p_{incorrect} <  \frac{2m}{s}
\end{equation}


\end{proof}

\section{Conclusion}

In this paper we present the first study of the delegation of quantum computing in a restricted model of computing and show that the general framework of the verification via blindness could be adapted to the setting of one-pure qubit model. In order to improve the obtained bound on the security parameter two open questions has to be addressed. The first one aims to expand the class of resource states for DQC1 model so that several techniques from the MBQC domain could be applicable here. The second question will complement the first by searching for fault-tolerant schemes based on any new resource state for DQC1 model. More concretely we propose the study of following questions:

\begin{itemize}

\item A sufficient condition for compatibility with DQC1 based on the step-wise determinism criteria is presented in Theorem \ref {t-compact}. Is this approach extendable to weaker notions of determinism such as information preserving maps as defined in \cite{fFlow10}? Which is a necessary condition for a family of MBQC resource states to be universal for the DQC1 computation?

\item Theorem \ref{verifiable_main} presents a scheme for verification where by adjusting the number of rounds one could obtain an $\epsilon$-verifiable delegated DQC1 computing with $\epsilon$ being inverse polynomial on computation size. How can we efficiently amplify this bound to any desired exponential one? Is there a way to adapt the proposed probability amplification method of \cite{fitzsimons2012unconditionally} based on a quantum error correcting code, into the DQC1 model?

\end{itemize}

%

\subparagraph*{Acknowledgements}

We would like to thank Joe Fitzsimons, Vedran Dunjko and Alistair Stewart for useful discussions. AD was supported in part by EPSRC (Grant No. EP/H03031X/1), U.S. EOARD (Grant No. 093020), and the EU Integrated Project SIQS. TK was supported by Mary and Armeane Choksi Postgraduate Scholarship and School of Informatics Graduate School.

\bibliographystyle{unsrt}
\bibliography{BlindQC}

\newcommand{\SortNoop}[1]{}
\begin{thebibliography}{10}

\bibitem{divincenzo00}
D.~DiVincenzo.
\newblock The physical implementation of quantum computation.
\newblock {\em Fortschr. Phys.}, 48:771, 2000.

\bibitem{BravyiKitaev02}
Sergey~B. Bravyi and Alexei~Yu. Kitaev.
\newblock Fermionic quantum computation.
\newblock {\em Annals of Physics}, 298:210, 2002.

\bibitem{BJS11}
M.~Bremner, R.~Jozsa, and D.~Shepherd.
\newblock Classical simulation of commuting quantum computations implies
  collapse of the polynomial hierarchy.
\newblock {\em Proc. Roy. Soc. A}, 467:459, 2011.

\bibitem{jordan10}
Stephen~P. Jordan.
\newblock Permutational quantum computing.
\newblock {\em Quantum Info. Comput.}, 10(5):470--497, May 2010.

\bibitem{AA11}
S.~Aaronson and A.~Arkhipov.
\newblock The computational complexity of linear optics.
\newblock In {\em STOC}, 2011.

\bibitem{Shor}
P.{\,}W. Shor.
\newblock Polynomial-time algorithms for prime factorization and discrete
  logarithms on a quantum computer.
\newblock {\em SIAM Journal on Computing}, 26:1484--1509, 1997.
\newblock First published in 1995.

\bibitem{Dorit}
D.~Aharonov, M.~{Ben-Or}, and E.~Eban.
\newblock Interactive proofs for quantum computations.
\newblock In {\em Proceedings of Innovations in Computer Science 2010}, page
  453, 2010.

\bibitem{fitzsimons2012unconditionally}
Joseph~F Fitzsimons and Elham Kashefi.
\newblock Unconditionally verifiable blind computation.
\newblock {\em arXiv preprint arXiv:1203.5217}, 2012.

\bibitem{RUV13}
B.~Reichardt F., Unger, and U.~Vazirani.
\newblock Classical command of quantum systems.
\newblock {\em Nature}, 496, 2013.

\bibitem{kl98}
E.~Knill and R.~Laflamme.
\newblock Power of one bit of quantum information.
\newblock {\em Phys. Rev. Lett.}, 81:5672, 1998.

\bibitem{s06}
D.~Shepherd.
\newblock Computing with unitaries and one pure qubit.
\newblock {\em arXiv:quant-ph/0608132}, 2006.

\bibitem{dattashaji11}
A.~Datta and A.~Shaji.
\newblock Quantum discord and quantum computing - an appraisal.
\newblock {\em International Journal of Quantum Information}, 9:1787, 2011.

\bibitem{morimae2013hardness}
Tomoyuki Morimae, Keisuke Fujii, and Joseph~F Fitzsimons.
\newblock On the hardness of classically simulating the one clean qubit model.
\newblock {\em arXiv preprint arXiv:1312.2496}, 2013.

\bibitem{RB01}
R.~Raussendorf and H.{\,}J. Briegel.
\newblock A one-way quantum computer.
\newblock {\em Physical Review Letters}, 86:5188 -- 5191, 2001.

\bibitem{BFK09}
A.~Broadbent, J.~Fitzsimons, and E.~Kashefi.
\newblock Universal blind quantum computing.
\newblock In {\em Proceedings of the 50th Annual IEEE Symposium on Foundations
  of Computer Science (FOCS 2009)}, page 517, 2009.

\bibitem{DKL11}
V.~Dunjko, E.~Kashefi, and A.~Leverrier.
\newblock Universal blind quantum computing with coherent states.
\newblock {\em arXiv preprint arXiv:1108.5571}, 2011.

\bibitem{MDK11}
T.~Morimae, V.~Dunjko, and E.~Kashefi.
\newblock Ground state blind quantum computation on aklt state.
\newblock {\em arXiv preprint arXiv:1009.3486}, 2011.

\bibitem{BKBFZW11}
S.~Barz, E.~Kashefi, A.~Broadbent, J.~F. Fitzsimons, A.~Zeilinger, and
  P.~Walther.
\newblock Demonstration of blind quantum computing.
\newblock {\em Science}, 335(6066):303--308, 2012.

\bibitem{FM12}
Tomoyuki Morimae and Keisuke Fujii.
\newblock Blind topological measurement-based quantum computation.
\newblock {\em Nature Communications}, 3:1036, 2012.

\bibitem{Morimae12}
Tomoyuki Morimae.
\newblock Continuous-variable blind quantum computation.
\newblock {\em Phys. Rev. Lett.}, 109:230502, Dec 2012.

\bibitem{MF12b}
Tomoyuki Morimae and Keisuke Fujii.
\newblock Blind quantum computation protocol in which alice only makes
  measurements.
\newblock {\em Phys. Rev. A}, 87:050301, May 2013.

\bibitem{SKM13}
T.~Sueki, T.~Koshiba, and T.~Morimae.
\newblock Ancilla-driven universal blind quantum computation.
\newblock {\em Physical Review A}, 87, 2013.

\bibitem{MPF13}
A.~Mantri, C.~Perez-Delgado, and J.~Fitzsimons.
\newblock Optimal blind quantum computation.
\newblock {\em arXiv preprint arXiv:1306.3677}, 2013.

\bibitem{GMMR13}
V.~Giovannetti, L.~Maccone, T.~Morimae, and T.~Rudolph.
\newblock Efficient universal blind computation.
\newblock {\em arXiv preprint arXiv:1306.2724}, 2013.

\bibitem{HEB04}
M.~Hein, J.~Eisert, and H.~J. Briegel.
\newblock Multi-party entanglement in graph states.
\newblock {\em Physical Review A}, 69, 2004.
\newblock quant-ph/0307130.

\bibitem{Mcal07}
V.~Danos, E.~Kashefi, and P.~Panangaden.
\newblock The measurement calculus.
\newblock {\em Journal of ACM}, 54:8, 2007.

\bibitem{Flow06}
V.~Danos and E.~Kashefi.
\newblock Determinism in the one-way model.
\newblock {\em Physical Review A}, 74:052310 [6 pages], 2006.

\bibitem{gflow07}
D.~Browne, E.~Kashefi, M.~Mhalla, and S.~Perdrix.
\newblock Generalized flow and determinism in measurement-based quantum
  computation.
\newblock {\em New Journal of Physics}, 9:250, 2007.

\bibitem{fFlow10}
M.~Mhalla, M.~Murao, S.~Perdrix, M.~Someya, and P.~Turner.
\newblock Which graph states are useful for quantum information processing?
\newblock In {\em TQC Theory of Quantum Computation, Communication and
  Cryptography 2011}, 2010.

\bibitem{DFPR13}
V.~Dunjko, J.~Fitzsimons, C.~Portmann R., and Renner.
\newblock Composable security of delegated quantum computation.
\newblock {\em arXiv preprint arXiv:1301.3662}, 2013.

\bibitem{BCGST02}
H.~Barnum, C.~Cr{\'e}peau, D.~Gottesman, A.~Smith, and A.~Tapp.
\newblock Authentication of quantum messages.
\newblock In {\em Proceedings of the 43rd Annual IEEE Symposium on Foundations
  of Computer Science (FOCS 2002)}, page 449, 2002.

\bibitem{barz2013experimental}
Stefanie Barz, Joseph~F Fitzsimons, Elham Kashefi, and Philip Walther.
\newblock Experimental verification of quantum computation.
\newblock {\em Nature Physics}, 2013.

\bibitem{pappa2012multipartite}
Anna Pappa, Andr{\'e} Chailloux, Stephanie Wehner, Eleni Diamanti, and Iordanis
  Kerenidis.
\newblock Multipartite entanglement verification resistant against dishonest
  parties.
\newblock {\em Physical Review Letters}, 108(26), 2012.

\bibitem{dujko2012thesis}
Vedran Dunjko.
\newblock Ideal quantum protocols in the non-ideal physical world.
\newblock {\em PhD Thesis, Heriot-Watt University}, 2012.

\end{thebibliography}

\appendix
\section{Proof of Theorem \ref{thm:blindness}} \label{appendix_1}

\begin{proof}

In this proof of blindness for Protocol  \ref{prot:blind} we use  techniques  developed in \cite{dujko2012thesis}. The basic difference from the proof of \cite{dujko2012thesis} arises from the different order in which Bob receives the states from Alice. Nevertheless, after commuting all CPTP maps into a single operator at the end, the methodology for proving blindness is the same as in the original proof. We give the full  proof here for the sake of clarity.

To prove blindness we do not separate Alice's system into a classical and a quantum part but we consider the whole of Alice's system as quantum. This is a reasonable assumption since a classical system can be viewed as a special case of a quantum system. Therefore, by proving blindness for the more general case we also prove blindness for the special case.

For the sake of clarity we use single indexing for all the qubits of the resource state. The total number of qubits is denoted by $m$ and the number of qubits in a single column of the brickwork state is denoted by $n$.

Our goal will be to explicitly write the state $\sigma_B= \text{Tr}_A (\sigma_{AB})$ that Bob holds at the end of the execution of the protocol. To achieve this we express Bob's behaviour at each step $i$ of the protocol as a collection of completely-positive trace-preserving (CPTP) maps $\mathcal{E}_{i}^{b_i}$, each for every possible classical response $b_i$ from Bob to Alice.

At step 1 of the main loop of the protocol Bob has already been given the top input qubit at position 1 (position $(1,1)$ in the protocol notation) and the qubit at position $f(1)=1+n$ (position $(1,2)$ in the protocol notation) together with the angle for measuring qubit 1 (angle can be represented as a quantum state composed of 3 qubits). State $\text{Tr}_A(\rho_{AB})$ represents Bob's state before the protocol begins and can, in general, be dependent on Alice's secret measurement angles. The state of Bob averaged over all possible choices of Alice and possible classical responses from Bob, after step 1 is:

\begin{equation*}
\sum_{b_{1},r_{1}, \theta_{1}, \theta_{1+n}} \mathcal{E}_{1}^{b_1} \left(  \ket{\delta_{1}^{\theta_{1},r_{1}}} \bra{\delta_{1}^{\theta_{1},r_{1}}} \otimes  \ket{+_{\theta_{1+n}}}\bra{+_{\theta_{1+n}}} \otimes  \ket{+_{\theta_{1}}}\bra{+_{\theta_{1}}} \otimes \text{Tr}_A(\rho_{AB}) \right)
\end{equation*}

Note the all binary parameters in sums range over 0 and 1, ex. $\sum_{b_{1}}$ stands for $\sum_{b_{1}=0}^{1}$ and all angles range over the 8 possible values in $A$.

We can write the state of Bob after step 2 of the main iteration as:

\begin{multline*}
\sum_{b_{2},b_{1},r_{2}, r_{1}, \theta_{2+n}, \theta_{1+n}, \theta_{2}, \theta_{1}} \mathcal{E}_{2}^{b_{2}} \Bigl(  \ket{\delta_{2}^{\theta_{2},r_{2}}}\bra{\delta_{2}^{\theta_{2},r_{2}}}  \otimes \ket{+_{\theta_{2+n}}}\bra{+_{\theta_{2+n}}}  \\ \otimes
 \mathcal{E}_{1}^{b_{1}} \Bigl( \ket{\delta_{1}^{\theta_{1},r_{1}}}\bra{\delta_{1}^{\theta_{1},r_{1}}}  \otimes  \ket{+_{\theta_{1+n}}}\bra{+_{\theta_{1+n}}} \otimes  \ket{+_{\theta_{1}}}\bra{+_{\theta_{1}}} \otimes \text{Tr}_A(\rho_{AB}) \Bigr) \Bigr)
\end{multline*}

Following this analysis, after the last step of the iteration Bob's state will be:

\begin{multline*}
\sigma_B = \sum_{\substack{\boldsymbol{b}_{\leq m-n},\\ \boldsymbol{r}_{\leq m-n}, \boldsymbol{\theta}_{\leq m}}} \mathcal{E}_{m-n}^{b_{m-n}} \Bigl(  \ket{\delta_{m-n}^{\boldsymbol{b}_{<m-n},\boldsymbol{r}_{\leq m-n},\theta_{m-n}}}\bra{\delta_{m-n}^{\boldsymbol{b}_{<m-n},\boldsymbol{r}_{\leq m-n},\theta_{m-n}}}  \otimes \ket{+_{\theta_{m}}}\bra{+_{\theta_{m}}}   \\
\otimes \ldots \otimes \mathcal{E}_{2}^{b_{2}} \Bigl( \ket{\delta_{2}^{\theta_{2},r_{2}}}\bra{\delta_{2}^{\theta_{2},r_{2}}}  \otimes \ket{+_{\theta_{2+n}}}\bra{+_{\theta_{2+n}}}  \\ \otimes
\mathcal{E}_{1}^{b_{1}} \Bigl( \ket{\delta_{1}^{\theta_{1},r_{1}}}\bra{\delta_{1}^{\theta_{1},r_{1}}}  \otimes  \ket{+_{\theta_{1+n}}}\bra{+_{\theta_{1+n}}} \otimes  \ket{+_{\theta_{1}}}\bra{+_{\theta_{1}}} \otimes \text{Tr}_A(\rho_{AB}) \Bigr) \Bigr) \ldots \Bigr)
\end{multline*}

Notation $\boldsymbol{b}_{<m-n}$ stands for all the elements of $\boldsymbol{b}$ with index less than $m-n$.

Collecting all CPTP maps by commuting them with systems which they do not apply on into a single operator $\mathcal{E}$ and rearranging the terms of the tensor product inside gives:

\begin{multline*}
=\sum_{\substack{\boldsymbol{b}_{\leq m-n},\\ \boldsymbol{r}_{\leq m-n}, \boldsymbol{\theta}_{\leq m}}} \mathcal{E}^{\boldsymbol{b}_{\leq m-n}} \Bigl( \bigotimes_{i=m-n}^{m} \ket{+_{\theta_i}}\bra{+_{\theta_i}}  \bigotimes_{i=n+1}^{m-n-1} (\ket{\delta_{i}^{\boldsymbol{b}_{<i},\boldsymbol{r}_{\leq i},\theta_{i}}}\bra{\delta_{i}^{\boldsymbol{b}_{<i},\boldsymbol{r}_{\leq i},\theta_{i}}}  \otimes \ket{+_{\theta_{i}}}\bra{+_{\theta_{i}}} ) \\
 \bigotimes_{i=2}^{n}( \ket{\delta_{i}^{\theta_i, r_i}}\bra{\delta_{i}^{\theta_i, r_i}} ) \otimes  \ket{\delta_{1}^{\theta_1, r_1}}\bra{\delta_{1}^{\theta_1, r_1}} \otimes  \ket{+_{\theta_{1}}}\bra{+_{\theta_{1}}} \otimes \text{Tr}_A(\rho_{AB})  \Bigr)
\end{multline*}

We introduce the controlled unitary:
$$ U =  \prod_{ n+1 \leq i \leq m-n-1, i=1} Z_i(-\delta_i)$$

and rewrite the state as:

\begin{multline*}
\sum_{\substack{\boldsymbol{b}_{\leq m-n},\\ \boldsymbol{r}_{\leq m-n}, \boldsymbol{\theta}_{\leq m}}} \mathcal{E}^{\boldsymbol{b}_{\leq m-n}} \Bigl(U^{\dagger} U \bigotimes_{i=m-n}^{m} \ket{+_{\theta_i}}\bra{+_{\theta_i}}  \bigotimes_{i=n+1}^{m-n-1} (\ket{\delta_{i}^{\boldsymbol{b}_{<i},\boldsymbol{r}_{\leq i},\theta_{i}}}\bra{\delta_{i}^{\boldsymbol{b}_{<i},\boldsymbol{r}_{\leq i},\theta_{i}}}  \otimes \ket{+_{\theta_{i}}}\bra{+_{\theta_{i}}} ) \\
 \bigotimes_{i=2}^{n}( \ket{\delta_{i}^{\theta_i, r_i}}\bra{\delta_{i}^{\theta_i, r_i}} ) \otimes  \ket{\delta_{1}^{\theta_1, r_1}}\bra{\delta_{1}^{\theta_1, r_1}} \otimes  \ket{+_{\theta_{1}}}\bra{+_{\theta_{1}}} U^{\dagger} U \otimes \text{Tr}_A(\rho_{AB})\Bigr)
\end{multline*}

After applying the innermost unitary and absorbing the outermost into the CPTP-map we have:

\begin{multline*}
\sum_{\substack{\boldsymbol{b}_{\leq m-n},\\ \boldsymbol{r}_{\leq m-n}, \boldsymbol{\theta}_{\leq m}}} \mathcal{E'}^{\boldsymbol{b}_{\leq m-n}} \Bigl( \bigotimes_{i=m-n}^{m} \ket{+_{\theta_i}}\bra{+_{\theta_i}} \\
 \bigotimes_{i=n+1}^{m-n-1} \left( \ket{\delta_{i}^{\boldsymbol{b}_{<i},\boldsymbol{r}_{\leq i},\theta_{i}}}\bra{\delta_{i}^{\boldsymbol{b}_{<i},\boldsymbol{r}_{\leq i},\theta_{i}}}  \otimes \ket{+_{-a_{i}^{'\text{  } \boldsymbol{b}_{<i},\boldsymbol{r}_{<i}} - r_i \pi }}\bra{+_{-a_{i}^{'\text{  } \boldsymbol{b}_{<i},\boldsymbol{r}_{<i}} - r_i \pi }} \right) \\
 \bigotimes_{i=2}^{n} \left( \ket{\delta_{i}^{\theta_i, r_i}}\bra{\delta_{i}^{\theta_i, r_i}} \right) \otimes  \ket{\delta_{1}^{\theta_1, r_1}}\bra{\delta_{1}^{\theta_1, r_1}} \otimes  \ket{+_{-a'_{1} - r_1 \pi }}\bra{+_{-a'_{1} - r_1 \pi }} \otimes \text{Tr}_A(\rho_{AB})\Bigr)
\end{multline*}

It is essential for the proof that each term with index $i$ in the tensor products depends only on parameters with index $\leq i$. This allows to break the summations over $\boldsymbol{r}_{\leq m-n}$ and $\boldsymbol{\theta}_{\leq m}$ and calculate them iteratively from left to right, given the following:

\begin{eqnarray}
 \sum_{\theta_i} \ket{+_{\theta_i}}\bra{+_{\theta_i}} = \frac{I_1}{2} \nonumber
\end{eqnarray}

where $I_n = \bigotimes_n I$. Also,

\begin{gather*}
\sum_{r_i, \theta_i} \ket{\delta_{i}^{\boldsymbol{r}_{\leq i},\theta_{i}}}\bra{\delta_{i}^{\boldsymbol{r}_{\leq i},\theta_{i}}}  \otimes \ket{+_{-a_{i}^{'\text{  } \boldsymbol{r}_{<i}} - r_i \pi }}\bra{+_{-a_{i}^{'\text{  } \boldsymbol{r}_{<i}} - r_i \pi }} \\
= \sum_{r_i} \Bigl( \sum_{\theta_{i}} \Bigl( \ket{a_{i}^{'\text{  } \boldsymbol{r}_{<i}} + \theta _{i} + r_{i} \pi}\bra{a_{i}^{'\text{  } \boldsymbol{r}_{<i}} + \theta _{i} + r_{i} \pi} \Bigr) \otimes \ket{+_{-a_{i}^{'\text{  } \boldsymbol{r}_{<i}}  - r_{i} \pi}}\bra{+_{-a_{i}^{'\text{  } \boldsymbol{r}_{<i}}  - r_{i} \pi}} \Bigr)  \\ =
 \sum_{r_{i}} \frac{I_3}{2^3} \otimes  \ket{+_{-a_{i}^{'\text{  } \boldsymbol{r}_{<i}}  - r_{i} \pi}}\bra{+_{-a_{i}^{'\text{  } \boldsymbol{r}_{<i}}  - r_{i} \pi}}     \\ =
  \frac{I_4}{2^4}
\end{gather*}

and

\begin{equation*}
\sum_{r_i, \theta_i} \ket{\delta_{i}^{\theta_i, r_i}}\bra{\delta_{i}^{\theta_i, r_i}} = \frac{I_3}{2^3}
\end{equation*}

This procedure will produce the state:

\begin{equation*}
\sigma_B = \mathcal{E}' \left( \frac{I_{4m-4n+1}}{2^{4m-4n+1}} \otimes \text{Tr}_A(\rho_{AB}) \right) = \mathcal{E}''( \text{Tr}_A(\rho_{AB}))
\end{equation*}

where $\mathcal{E}''$ is some CPTP map. Therefore Definition \ref{def:blindness} is satisfied.

\end{proof}

\section{Proof of Theorem \ref{verifiable_main}} \label{appendix_2}

\begin{proof}

The same notation is used as in Section \ref{section_3}. The first step is to write the state of Alice's system at the end of the execution of the protocol for fixed Bob's behaviour $j$ and choices of Alice $\nu$. We have utilised the fact that all measurements can be moved to the end. Also, we have commuted all Bob's operations to the end (before the measurements) merging them to a single CPTP map. The state of Alice is:

\begin{multline}
B_j(\nu)  =  \sum_{\boldsymbol{b}}  \otimes_{i=k}^{s} \ket{+_{\theta_{t_k}+b_{t_k} \pi}} \bra{+_{\theta_{t_k}+b_{t_k} \pi}} C^{\boldsymbol{b},\nu_C} \ket{\boldsymbol{b} + \boldsymbol{c}^{\boldsymbol{r}}}\bra{\boldsymbol{b}} \nonumber \\
\mathcal{E}^{j} \left( \bigotimes_{k=1}^{s} \bigotimes_{i=1}^{m-n} \ket{\delta_{(k-1)m+i}^{\boldsymbol{b},\nu}}\bra{\delta_{(k-1)m+i}^{\boldsymbol{b},\nu}}  \otimes \rho_{M^{\nu}_k}  \right)
\ket{\boldsymbol{b}}\bra{\boldsymbol{b} + \boldsymbol{c}^{\boldsymbol{r}}} C^{\boldsymbol{b},\nu_C \dagger}  \otimes_{i=k}^{s} \ket{+_{\theta_{t_k}+b_{t_k} \pi}} \bra{+_{\theta_{t_k}+b_{t_k} \pi}} \nonumber
\end{multline}

where $\ket{+_{\theta_{t_k}+b_{t_k} \pi}} \bra{+_{\theta_{t_k}+b_{t_k} \pi}}$ are used to define Alice's measurement of the traps which are part of the output state of each round $k$ (if they exist).

To bound the failure probability, observe that projectors orthogonal to $\ket{\eta_{t_k}^{\nu_T}}$'s vanish, thus we have (where $b'=\{b_i\}_{i \neq t_1 \ldots t_s}$):

\begin{multline}
p_{incorrect}  =   \sum_{\boldsymbol{b'},\nu} p(\nu) \text{Tr} (P_{\perp}  \bigotimes_{k=1}^{s} \ket{\eta_{t_k}^{\nu_T}} \bra{\eta_{t_k}^{\nu_T}} \nonumber \\
C^{\boldsymbol{b'},\nu_C} \ket{\boldsymbol{b'} + \boldsymbol{c}^{\boldsymbol{r}}}\bra{\boldsymbol{b'}}\mathcal{E}^{j} \left( \bigotimes_{k=1}^{s} \bigotimes_{i=1}^{m-n} \ket{\delta_{(k-1)m+i}^{\boldsymbol{b'},\nu}}\bra{\delta_{(k-1)m+i}^{\boldsymbol{b'},\nu}}  \otimes \rho_{M^{\nu}_k}  \right)   \ket{\boldsymbol{b'}}\bra{\boldsymbol{b'} + \boldsymbol{c}^{\boldsymbol{r}}} C^{\boldsymbol{b'},\nu_C \dagger}) \nonumber
\end{multline}

We introduce the following unitary, which characterises the correct operation on each subgraph $k$: $\mathcal{P}_k = \bigotimes_{i=1}^{m-n} (H_{(k-1)m+i} Z_{(k-1)m+i}(\delta_{(k-1)m+i})) E_{G'_k} $.

We can rewrite the failure probability, introducing $\mathcal{P}^{\dagger}_k \mathcal{P}_k$'s on both sides of the quantum state of the system before the attack, and absorbing the outermost unitaries into the updated CPTP map $\mathcal{E'}^j$:

\begin{multline}
p_{incorrect}  =   \sum_{\boldsymbol{b'},\nu} p(\nu) \text{Tr} (P_{\perp}  \bigotimes_{k=1}^{s} \ket{\eta_{t_k}^{\nu_T}} \bra{\eta_{t_k}^{\nu_T}} C^{\boldsymbol{b'},\nu_C}\nonumber \\
 \ket{\boldsymbol{b'} + \boldsymbol{c}^{\boldsymbol{r}}}\bra{\boldsymbol{b'}}\mathcal{E'}^{j} \left( \bigotimes_{k=1}^{s} (\mathcal{P}_k \bigotimes_{i=1}^{m-n} \ket{\delta_{(k-1)m+i}^{\boldsymbol{b'},\nu}}\bra{\delta_{(k-1)m+i}^{\boldsymbol{b'},\nu}}  \otimes \rho_{M^{\nu}_k}  \mathcal{P}^{\dagger}_k) \right)  \ket{\boldsymbol{b'}}\bra{\boldsymbol{b'} + \boldsymbol{c}^{\boldsymbol{r}}} C^{\boldsymbol{b'},\nu_C \dagger}) \nonumber
\end{multline}

We decompose $\mathcal{E'}^{j}$ using the following facts: There exist some matrices $\{ \chi_v \}$ of dimension $s(4m-3n) \times s(4m-3n)$, with $\sum_v \chi_v \chi_v^{\dagger} = I$ such that for every density operator $\rho$: $\mathcal{E'}^{j}(\rho) = \sum_v \chi_v \rho \chi_v^{\dagger}$. Also, each $\chi_v$ can be decomposed to the Pauli basis: $\chi_v = \sum_i \alpha_{vi}\sigma_i$, with $\sum_{v,i} \alpha_{vi} \alpha_{vi}^* = 1$. Setting $\sigma_{i,k}$ to be the part of $\sigma_{i}$ that applies on the qubits $(k-1)m+1 \leq \gamma \leq km$.

\begin{multline}
p_{incorrect}  =   \sum_{\boldsymbol{b'},\nu, v , i, j} \alpha_{vi} \alpha_{vj}^* p(\nu) \text{Tr} (P_{\perp}  \bigotimes_{k=1}^{s} \ket{\eta_{t_k}^{\nu_T}} \bra{\eta_{t_k}^{\nu_T}} C^{\boldsymbol{b'},\nu_C} \nonumber \\
 \ket{\boldsymbol{b'} + \boldsymbol{c}^{\boldsymbol{r}}}\bra{\boldsymbol{b'}} \bigotimes_{k=1}^{s} ( \sigma_{i,k} \left( \mathcal{P}_k  \bigotimes_{i=1}^{m-n} \ket{\delta_{(k-1)m+i}^{\boldsymbol{b'},\nu}}\bra{\delta_{(k-1)m+i}^{\boldsymbol{b'},\nu}}  \otimes \rho_{M^{\nu}_k}  \mathcal{P}^{\dagger}_k \right) \sigma_{j,k}) \ket{\boldsymbol{b'}}\bra{\boldsymbol{b'} + \boldsymbol{c}^{\boldsymbol{r}}} C^{\boldsymbol{b'},\nu_C \dagger} \nonumber
\end{multline}

Without loss of generality we can assume that $\sigma_i$, $\sigma_j$ do not change the $\delta$'s.

For an arbitrary $t_g$, the only attacks that give the corresponding term of the sum not equal to zero:

\begin{multline*}
P_{\perp} ( C^{\boldsymbol{b'},\nu_C}  \ket{\boldsymbol{b'}}\bra{\boldsymbol{b'} + \boldsymbol{c}^{\boldsymbol{r}}} \sigma_{i,t_g} \\ ( \mathcal{P}_{t_g}  \bigotimes_{i=1}^{m-n} \ket{\delta_{(t_g-1)m+i}^{\boldsymbol{b'},\nu}}\bra{\delta_{(t_g-1)m+i}^{\boldsymbol{b'},\nu}}  \otimes \rho_{M^{\nu}_{t_g}}  \mathcal{P}^{\dagger}_{t_g} ) \sigma_{j,t_g}  \ket{\boldsymbol{b'}}\bra{\boldsymbol{b'} + \boldsymbol{c}^{\boldsymbol{r}}} C^{\boldsymbol{b'},\nu_C \dagger} )\neq 0
\end{multline*}

are those that (i) produce an incorrect measurement result for qubits $(t_g-1) m +1 \leq \gamma \leq t_gm-n$ or (ii) operate non-trivially on qubits $ t_g m - n < \gamma \leq t_g m$. We denote this condition by $i\in E_{i,t_g}$ and $j\in E_{j,t_g}$.

We can rewrite the probability by eliminating $P_{\perp}$ (observing that it applies to a positive operator) and $C^{\boldsymbol{b'},\nu_C}$ (by the cyclical property of the trace):

\begin{multline}
p_{incorrect}  \leq   \sum_{\nu, v , i \in E_{i,t_g}, j \in E_{j,t_g}}  \alpha_{vi} \alpha_{vj}^* p(\nu) \prod_{k=1}^{s} \text{Tr} (  \ket{\eta_{t_k}^{\nu_T}} \bra{\eta_{t_k}^{\nu_T}} \nonumber \\
\ket{\boldsymbol{b'}}\bra{\boldsymbol{b'} + \boldsymbol{c}^{\boldsymbol{r}}} \sigma_{i,k} \left( \mathcal{P}_k  \bigotimes_{i=1}^{m-n} \ket{\delta_{(k-1)m+i}^{\boldsymbol{b'},\nu}}\bra{\delta_{(k-1)m+i}^{\boldsymbol{b'},\nu}}  \otimes \rho_{M^{\nu}_k}  \mathcal{P}^{\dagger}_k \right) \sigma_{j,k}) \nonumber
\end{multline}

We extract a trace over R from $\rho_{M^{\nu}_{t_g}}$. And extract the sums over $\nu_{C,k}$'s from the general sum, where $\nu_{C,k}$ is the subset of random parameters $\nu_C$ that are used for the computation of round $r$:

\begin{multline}
=   \sum_{\nu_T, v , i \in E_{i,t_g}, j \in E_{j,t_g}}  \alpha_{vi} \alpha_{vj}^* p(\nu_T) \prod_{k=1}^{s} \text{Tr} (  \ket{\eta_{t_k}^{\nu_T}} \bra{\eta_{t_k}^{\nu_T}} \nonumber \\
\ket{\boldsymbol{b'}}\bra{\boldsymbol{b'} + \boldsymbol{c}^{\boldsymbol{r}}} \sigma_{i,k} \left( \mathcal{P}_k  \sum_{\nu_{C,k}} (p(\nu_{C,k})  \bigotimes_{i=1}^{m-n} \ket{\delta_{(k-1)m+i}^{\boldsymbol{b'},\nu}}\bra{\delta_{(k-1)m+i}^{\boldsymbol{b'},\nu}}  \otimes \text{Tr}_R (\rho_{M^{\nu}_k} ) ) \mathcal{P}^{\dagger}_k \right) \sigma_{j,k}) \nonumber
\end{multline}

To take advantage of the blindness property we use the following lemma where the proof is given later.

\begin{lemma} [Blindness (excluding the traps)] \label{usefullemma_second}

\begin{gather}
 \forall k, \sum_{\nu_{C,k}} p(\nu_{C,k})   \bigotimes_{i=1}^{m-n} \ket{\delta_{(k-1)m+i}^{\boldsymbol{b'},\nu}}\bra{\delta_{(k-1)m+i}^{\boldsymbol{b'},\nu}}  \otimes \text{Tr}_R (\rho_{M^{\nu}_k} ) \nonumber \\
  =  \frac{I^{t_k}_k}{\text{Tr}(I^{t_k}_k)} \otimes  \ket{\delta_{t_k}^{\theta_{t_k}, r_{t_k}}} \bra{\delta_{t_k}^{\theta_{t_k}, r_{t_k}}} \otimes \ket{+_{\theta _{t_k}}} \bra{+_{\theta _{t_k}}} \nonumber
\end{gather}

If $k \neq t_g$, $I^{t_k}_k = \bigotimes_{4m-3n-1} I$ when $km-n < t_k \leq km $ and $I^{t_k}_k = \bigotimes_{4m-3n-4} I$ when $(k-1)m < t_k \leq km-n$ . And if $k=t_g$, $I^{t_k}_k = \bigotimes_{4m-3n} I$.
\end{lemma}

Lemma \ref{usefullemma_second} allows us to simplify the big sum above based on the position of the traps. We also sum over $\boldsymbol{b'}$ since there are no longer any dependencies on it in the sum, obtaining:

\begin{multline}
= \sum_{t_g, v, i\in E_{i,t_g},j\in E_{j,t_g}} \alpha_{vi} \alpha_{vj}^* p(t_g) \prod_{k=1}^{s} \text{Tr}  ( \nonumber \\
\sum_{\substack{km-n < t_k \leq km,\\ \theta_{t_k}}} p(t_k, \theta_{t_k}) \ket{+_{\theta_{t_k}}} \bra{+_{\theta_{t_k}}}  \sigma_{i,k}( \frac{\mathcal{I}}{\text{Tr}(\mathcal{I})}  \otimes \ket{+_{\theta_{t_k}}} \bra{+_{\theta_{t_k}}}   ) \sigma_{j,k} \nonumber \\
+ \sum_{\substack{(k-1)m < t_k \leq km-n,\\  r_{t_k}}} p(t_k,  r_{t_k}) \ket{r_{t_k}} \bra{r_{t_k}}  \sigma_{i,k}(\frac{\mathcal{I}}{\text{Tr}(\mathcal{I})} \otimes  \ket{r_{t_k}} \bra{r_{t_k}}  ) \sigma_{j,k}) \nonumber
\end{multline}

where $\mathcal{I}=\bigotimes_{4m-3n-1} I$ when $k \neq t_g$. And $\mathcal{I} = \bigotimes_{4m-3n} I$ when $k=t_g$.

Note that $\sum_{\theta_{t_k}} \text{Tr} (\ket{+_{\theta_{t_k}}} \bra{+_{\theta_{t_k}}}  \sigma_{i,k} ( \frac{\mathcal{I}}{\text{Tr}(\mathcal{I})} \otimes  \ket{+_{\theta_{t_k}}} \bra{+_{\theta_{t_k}}} ) \sigma_{j,k})$ is zero if $\sigma_{i,k} \neq \sigma_{j,k}$. The same is true for $\sum_{r_{t_k}} \text{Tr} (\ket{r_{t_k}} \bra{r_{t_k}}  \sigma_{i,k}( \frac{\mathcal{I}}{\text{Tr}(\mathcal{I})} \otimes \ket{r_{t_k}} \bra{r_{t_k}} ) \sigma_{j,k})$. Therefore we can only keep those terms where $\sigma_{i,k} = \sigma_{j,k}$ and the failure probability becomes:

\begin{multline}
= \sum_{t_g} \sum_{v, i\in E_{i,t_g}} | \alpha_{vi} |^2 p(t_g) \prod_{k=\{1, \ldots, s\} \setminus t_g}  ( \sum_{\substack{km-n < t_k \leq km,\\ \theta_{t_k}}} p(t_k, \theta_{t_k}) ( \bra{+_{\theta_{t_k}}}  \sigma_{i|t_k} \ket{+_{\theta_{t_k}}}   )^2 \nonumber \\
+ \sum_{\substack{(k-1)m < t_k \leq km-n,\\  r_{t_k}}}  p(t_k,  r_{t_k}) ( \bra{r_{t_k}}  \sigma_{i|t_k}\ket{r_{t_k}} )^2 ) \nonumber
\end{multline}

The rest of the proof is based on a counting argument. For convenience we introduce the following sets for an arbitrary Pauli $\sigma_{i,k}$:

\begin{gather}
A_{i,k} = \{ \gamma \text{ s.t. } \sigma_{i|\gamma}=I \text{ and } (k-1)m+1 \leq \gamma \leq km \} \nonumber \\
B_{i,k} = \{ \gamma \text{ s.t. } \sigma_{i|\gamma}=X \text{ and } (k-1)m+1 \leq \gamma \leq km \} \nonumber \\
C_{i,k} = \{ \gamma \text{ s.t. } \sigma_{i|\gamma}=Y \text{ and } (k-1)m+1 \leq \gamma \leq km \} \nonumber \\
D_{i,k} = \{ \gamma \text{ s.t. } \sigma_{i|\gamma}=Z \text{ and } (k-1)m+1 \leq \gamma \leq km \} \nonumber
\end{gather}

and use the superscript $O$ to denote subsets subject to the constraint $km \geq \gamma \geq km-n+1$.

The failure probability is then:

\begin{multline}
= \sum_{t_g} \sum_{v, i\in E_{i,t_g}} | \alpha_{vi} |^2 \frac{1}{s} \prod_{k=\{1, \ldots, s\} \setminus t_g} ( ( \frac{1}{8m} (8 |A_{i,k}^O |+ 4 |B_{i,k}^O| + 4 |C_{i,k}^O|  )  + \nonumber \\
\frac{1}{2m} (2 |A_{i,k} \setminus A_{i,k}^O| + 2 | D_{i,k} \setminus D_{i,k}^O| ) )  \nonumber
\end{multline}

Merging the terms:

\begin{equation}
= \sum_{t_g} \sum_{v, i\in E_{i,t_g}} | \alpha_{vi} |^2 \frac{1}{s} \prod_{k=\{1, \ldots, s\} \setminus t_g} \frac{1}{2m}(  2 |A_{i,k}| + |B_{i,k}^O| +|C_{i,k}^O|+ 2 |D_{i,k} \setminus D_{i,k}^O |)  \nonumber
\end{equation}

Using the fact that for every $k$, $|A_{i,k}|+|B_{i,k}|+|C_{i,k}|+|D_{i,k}|=m$:

\begin{equation}
\leq  \sum_{t_g} \sum_{v, i\in E_{i,t_g}} | \alpha_{vi} |^2 \frac{1}{s} \prod_{k=\{1, \ldots, s\} \setminus t_g} \frac{1}{2m}(  2m - |B_{i,k}| - |C_{i,k}| - |D_{i,k}^O|) \nonumber
\end{equation}

The conditions $i\in E_{i,t_g}$ that we obtained at the first part of the proof are translated to $|B_{i,t_g}|+|C_{i,t_g}|+|D_{i,t_g}^O| \geq 1$. In order to be able to use these conditions we need to rewrite the formula. First we expand it:

\begin{multline}
= \frac{1}{s}( \sum_{v, i\in E_{i,1}} | \alpha_{vi} |^2  \prod_{k=\{2, 3, \ldots, s\}} \frac{1}{2m} (  2m - |B_{i,k}| - |C_{i,k}| - |D_{i,k}^O|) \nonumber \\
+ \sum_{v, i\in E_{i,2}} | \alpha_{vi} |^2  \prod_{k=\{1, 3, 4, \ldots, s\}} \frac{1}{2m}(  2m - |B_{i,k}| - |C_{i,k}| - |D_{i,k}^O|) \nonumber \\
\ldots + \sum_{v, i\in E_{i,d}} | \alpha_{vi} |^2  \prod_{k=\{1, 2, \ldots, s-1\}} \frac{1}{2m} (  2m - |B_{i,k}| - |C_{i,k}| - |D_{i,k}^O|)) \nonumber
\end{multline}

We denote the product term $\prod_{k=\{1, 2, 3, \ldots, s\}\setminus z} \frac{1}{2m}(  2m - |B_{i,k}| - |C_{i,k}| - |D_{i,k}^O|)$ as $P_{i,z}$. We also denote each set $\{E_{i,1}^* \cap E_{i,2}^* \cap \ldots \cap E_{i,s}^*\}$, where each term $E_{i,w}^*$ is either $E_{i,w}$ or its complement, $E_{i,w}^C$, depending on whether the $w$-th value of a binary vector $\boldsymbol{y}$ (size $s$) is 1 or 0 respectively, as $W_{i,\boldsymbol{y}}$. Then we have:

\begin{equation}
= \frac{1}{s} ( \sum_{\boldsymbol{y}\setminus (0 \ldots 0)} \sum_{i \in W_{i,\boldsymbol{y}},v} ( | \alpha_{vi} |^2 \sum_{\{z:y_z=1\}} P_{i,z})) \nonumber
\end{equation}

Let the function $\# \boldsymbol{y}$ give the number of positions $i$ such that $y_i$=1.

\begin{equation}
= \frac{1}{s} ( \sum_{k=1}^s \sum_{\{\boldsymbol{y} : \#\boldsymbol{y}=k\}} \sum_{i \in W_{i,\boldsymbol{y}},v} ( | \alpha_{vi} |^2 \sum_{\{z:y_z=1\}} P_{i,z})) \nonumber
\end{equation}

We separately consider the following term for any arbitrary $\boldsymbol{y}$ with $\#\boldsymbol{y}=r$.

\begin{equation}
 \sum_{i \in W_{i,\boldsymbol{y}}} ( | \alpha_{vi} |^2 \sum_{\{z:y_z=1\}} P_{i,z}) \nonumber
\end{equation}

The condition $ i \in W_{i,\boldsymbol{y}}$ means that the following conditions hold together: $\{ |B_{i,w}|+|C_{i,w}|+|D_{i,w}^O| \geq 1 : y_w=1\}$,$\{ |B_{i,w}|+|C_{i,w}|+|D_{i,w}^O| = 0 : y_w=0\}$. We expand:

\begin{gather}
= \sum_{i \in W_{i,\boldsymbol{y}}} ( | \alpha_{vi} |^2 \sum_{\{z:y_z=1\}} \prod_{k=\{1, 2, 3, \ldots, s\}\setminus z} \frac{1}{2m}(  2m - |B_{i,k}| - |C_{i,k}| - |D_{i,k}^O|) \nonumber \\
= \sum_{i \in W_{i,\boldsymbol{y}}} ( | \alpha_{vi} |^2 \sum_{\{z:y_z=1\}} \prod_{\{k: y_k=1, k \neq z\}} \frac{1}{2m}(  2m - |B_{i,k}| - |C_{i,k}| - |D_{i,k}^O|) \nonumber \\ \prod_{\{k: y_k=0\}} \frac{1}{2m}(  2m - |B_{i,k}| - |C_{i,k}| - |D_{i,k}^O|) \nonumber
\end{gather}

And by using the above conditions:

\begin{gather}
\leq \sum_{i \in W_{i,\boldsymbol{y}}} ( | \alpha_{vi} |^2 \sum_{\{z:y_z=1\}} \prod_{\{k: y_k=1, k \neq z\}} \frac{1}{2m}(  2m - 1)  \prod_{\{k: y_k=0\}} \frac{1}{2m}(  2m ) \nonumber \\
= \sum_{i \in W_{i,\boldsymbol{y}}} ( | \alpha_{vi} |^2 \sum_{\{z:y_z=1\}} \left(\frac{2m-1}{2m}\right)^{r-1} \nonumber \\
= \sum_{i \in W_{i,\boldsymbol{y}}}  | \alpha_{vi} |^2 r \left(\frac{2m-1}{2m}\right)^{r-1} \nonumber
\end{gather}

Therefore the bound of our failure probability will be:

\begin{align}
p_{\text{incorrect}} &\leq  \frac{1}{s} ( \sum_{k=1}^s \sum_{\{\boldsymbol{y} : \#\boldsymbol{y}=k\}} \sum_{i \in W_{i,\boldsymbol{y}},v}  | \alpha_{vi} |^2 k \left(\frac{2m-1}{2m}\right)^{k-1})  \nonumber \\
&= \frac{1}{s} ( \sum_{k=1}^s k \left(\frac{2m-1}{2m}\right)^{k-1} \sum_{\{\boldsymbol{y} : \#\boldsymbol{y}=k\}} \sum_{i \in W_{i,\boldsymbol{y}},v}  | \alpha_{vi} |^2 )  \nonumber \\
&= \frac{1}{s} ( \sum_{k=1}^s c_k k \left(\frac{2m-1}{2m}\right)^{k-1} ) \nonumber
\end{align}

where $c_k = \sum_{\{\boldsymbol{y} : \#\boldsymbol{y}=k\}} \sum_{i \in W_{i,\boldsymbol{y}},v}  | \alpha_{vi} |^2$

subject to conditions:

\begin{equation}
\sum_{k=1}^{s} c_k \leq 1
\end{equation}

and

\begin{equation}
\forall k: c_k \geq 0
\end{equation}

\end{proof}

\begin{proof} [Proof of Lemma \ref{usefullemma_second}]

First we define state $\ket{q_i}$ as:

\begin{align*}
i \in D & &  \ket{q_i} &\equiv\ket{d_i} \\
i \notin D  & & \ket{q_i}&\equiv(\prod_{\{j : j \sim i, j \in D \}} Z^{d_j}) \ket{+_{\theta _{i}}}
\end{align*}

By substituting $\rho_{M^{\nu}_k}$'s and taking the trace over R:

If $k \neq t_g$ the state becomes:

\begin{multline}
\sum_{\nu_{C,k}} p(\nu_{C,k}) ( \bigotimes_{i=km-n+1}^{km} \ket{q_i}\bra{q_i} \bigotimes_{i=(k-1)m+n+1}^{km-n} \left( \ket{\delta_{i}^{\boldsymbol{b'},\nu}}\bra{\delta_{i}^{\boldsymbol{b'},\nu}} \otimes \ket{q_i^{\nu}}\bra{q_i^{\nu}} \right) \nonumber \\
 \bigotimes_{i=1}^{2} \left( \ket{\delta_{p_{i,k}}^{\boldsymbol{b'},\nu}}\bra{\delta_{p_{i,k}}^{\boldsymbol{b'},\nu}} \otimes \ket{q_{p_{i,k}}^{\nu}} \bra{q_{p_{i,k}}^{\nu}} \right)  \otimes I_{4(n-2)}/2^{4(n-2)} ) \nonumber
\end{multline}

where $\ket{q_{p_{i,k}}^{\nu}}$ denote the first layer pure qubits (a maximum of two) of the $k$-th graph state, used as padding (dummies) or trap and their positions are defined as: $1+(k-1)m \leq \{ p_{1,k}, p_{2,k} \} \leq n+(k-1)m$.

Otherwise, if $k = t_g$ the state becomes:

\begin{multline}
\sum_{\nu_{C,k}} p(\nu_{C,k})  ( \bigotimes_{i=t_g m-n+1}^{t_g m} \ket{q_i}\bra{q_i} \bigotimes_{i=(t_g -1)m+n+1}^{t_g m-n} \left( \ket{\delta_{i}^{\boldsymbol{b'},\nu}}\bra{\delta_{i}^{\boldsymbol{b'},\nu}} \otimes \ket{q_i^{\nu}}\bra{q_i^{\nu}} \right) \nonumber \\
  \otimes \ket{\delta_{u}^{\theta_{u}, r_{u}}}\bra{\delta_{u}^{\theta_{u}, r_{u}}} \otimes \ket{q_{u}^{\theta_{u}}}\bra{q_{u}^{\theta_{u}}} \otimes I_{4(w-1)}/2^{4(w-1)} ) \nonumber
\end{multline}

where $u=(t_g-1)m+1$ is the position of the single pure qubit of the input to the DQC1-MBQC computation.

An implicit assumption was that all $\delta$'s that are used to implement the measurements of maximally mixed inputs are maximally mixed states themselves, without any loss of generality.

We define a new controlled unitary:

\begin{equation}
\mathcal{P'}_k =  \left( \prod_{\{i : i \notin D, (k-1)m + 1 \leq i \leq km-n\}} Z_i(-\delta_i)  \right) \prod_{\{i : i \notin D_k\}} \prod_{\{j : j \sim i, j \in D_k \}} Z_i(d_j)
\end{equation}

where $D_k$ denotes the set of dummies of subgraph $G'_k$.

Using this unitary we rewrite the state. If $k \neq t_g$ it becomes:

\begin{multline}
\sum_{\nu_{C_k}} p(\nu_{C,k})  \mathcal{P'}^{\dagger} \mathcal{P'}  ( \bigotimes_{i=km-n+1}^{km} \ket{q_i}\bra{q_i} \bigotimes_{i=(k-1)m+n+1}^{km-n} \left( \ket{\delta_{i}^{\boldsymbol{b'},\nu}}\bra{\delta_{i}^{\boldsymbol{b'},\nu}} \otimes \ket{q_i^{\nu}}\bra{q_i^{\nu}} \right) \nonumber \\
 \bigotimes_{i=1}^{2} \left( \ket{\delta_{p_{i,k}}^{\boldsymbol{b'},\nu}}\bra{\delta_{p_{i,k}}^{\boldsymbol{b'},\nu}} \otimes \ket{q_{p_{i,k}}^{\nu}} \bra{q_{p_{i,k}}^{\nu}} \right)  \otimes I_{4(n-2)}/2^{4(n-2)} ) \mathcal{P'}^{\dagger} \mathcal{P'} \nonumber
\end{multline}

Otherwise:

\begin{multline}
\sum_{\nu_{C,k}} p(\nu_{C,k})  \mathcal{P'}^{\dagger} \mathcal{P'}  \bigotimes_{i=t_g m-n+1}^{t_g m} \ket{q_i}\bra{q_i} \bigotimes_{i=(t_g -1)m+n+1}^{t_g m-n} \left( \ket{\delta_{i}^{\boldsymbol{b'},\nu}}\bra{\delta_{i}^{\boldsymbol{b'},\nu}} \otimes \ket{q_i^{\nu}}\bra{q_i^{\nu}} \right) \nonumber \\
   \otimes \ket{\delta_{u}^{\theta_{u}, r_{u}}}\bra{\delta_{u}^{\theta_{u}, r_{u}}} \otimes \ket{q_{u}^{\theta_{u}}}\bra{q_{u}^{\theta_{u}}} \otimes I_{4(w-1)}/2^{4(w-1)} )  \mathcal{P'}^{\dagger} \mathcal{P'} \nonumber
\end{multline}

After applying the innermost unitary, if $k \neq t_g$:

\begin{multline}
\sum_{\nu_{C,k}} p(\nu_{C,k})  \mathcal{P'}^{\dagger} ( \bigotimes_{i=km-n+1}^{km} \ket{q'_i}\bra{q'_i} \bigotimes_{i=(k-1)m+n+1}^{km-n} \left( \ket{\delta_{i}^{\boldsymbol{b'},\nu}}\bra{\delta_{i}^{\boldsymbol{b'},\nu}} \otimes \ket{q_i^{'\nu}}\bra{q_i^{'\nu}} \right) \nonumber \\
 \bigotimes_{i=1}^{2} \left( \ket{\delta_{p_{i,k}}^{\boldsymbol{b'},\nu}}\bra{\delta_{p_{i,k}}^{\boldsymbol{b'},\nu}} \otimes \ket{q_{p_{i,k}}^{'\nu}} \bra{q_{p_{i,k}}^{'\nu}} \right)  \otimes I_{4(n-2)}/2^{4(n-2)} )  \mathcal{P'} \nonumber
\end{multline}

where state $\ket{q'_i}$ is defined as:

\begin{align*}
i \in D &  & \ket{q'_i}&\equiv\ket{d_i} \\
i \notin D, \forall k: k m \geq i \geq k m-n+1 & & \ket{q'_i}&\equiv \ket{+_{\theta _{i}}} \\
i \notin D, \forall k: k m-n \geq i \geq (k-1)m+1 &  & \ket{q'_i}&\equiv \ket{+_{-a_{i}^{''\text{  } \boldsymbol{b'},\boldsymbol{r}_{<i}} - r_i \pi }} \\
\end{align*}

Otherwise, if $k = t_g$:

\begin{multline}
\sum_{\nu_{C,k}} p(\nu_{C,k})  \mathcal{P'}^{\dagger} ( \bigotimes_{i=t_g m-n+1}^{t_g m} \ket{q'_i}\bra{q'_i} \bigotimes_{i=(t_g -1)m+n+1}^{t_g m-n} \left( \ket{\delta_{i}^{\boldsymbol{b'},\nu}}\bra{\delta_{i}^{\boldsymbol{b'},\nu}} \otimes \ket{q_i^{'\nu}}\bra{q_i^{'\nu}} \right) \nonumber \\
   \otimes \ket{\delta_{u}^{\theta_{u}, r_{u}}}\bra{\delta_{u}^{\theta_{u}, r_{u}}} \otimes \ket{q_{u}^{' \theta_{u}}}\bra{q_{u}^{' \theta_{u}}} \otimes I_{4(w-1)}/2^{4(w-1)} )  \mathcal{P'} \nonumber
\end{multline}

It is essential for the proof that each term with index $i$ in the tensor product depends only on parameters with index $\leq i$ and the term with index $(t_g-1)m+1$ (input qubit) and the trap qubit and its  measurement angle (if it is not an output) depend only on their own parameters. This allows to break the summations and calculate them iteratively from left to right, given the following:

\begin{equation}
 \sum_{d_i} p(d_i) \ket{d_i}\bra{d_i} = \frac{I}{2} \nonumber
\end{equation}

\begin{equation}
 \sum_{\theta _{i}} p(\theta _{i})  \ket{+_{\theta _{i}}} \bra{+_{\theta _{i}}} =  \frac{I}{2} \nonumber
\end{equation}

\begin{equation}
 \sum_{\theta _{i},r_i, d_i} p(\theta _{i},r_i, d_i)  \ket{\delta_{i}^{\boldsymbol{b'},\nu}}\bra{\delta_{i}^{\boldsymbol{b'},\nu}} \otimes \ket{d_i}\bra{d_i} = \frac{I_4}{2^4} \nonumber
\end{equation}

\begin{gather}
 \sum_{\theta _{i},r_i} p(\theta _{i},r_i)  \ket{\delta_{i}^{\boldsymbol{b'},\nu}}\bra{\delta_{i}^{\boldsymbol{b'},\nu}} \otimes \ket{+_{-a_{i}^{''\text{  } \boldsymbol{b'},\boldsymbol{r}_{<i}} - r_i \pi }}\bra{+_{-a_{i}^{''\text{  } \boldsymbol{b'},\boldsymbol{r}_{<i}} - r_i \pi }}  \nonumber \\
\begin{split}
= \sum_{r_i} p(r_i) \left( \sum_{\theta _{i}} p(\theta _{i})  \ket{a_{i}^{''\text{  } \boldsymbol{b'},\boldsymbol{r}_{<i}} + \theta_i + r_i \pi }\bra{a_{i}^{''\text{  } \boldsymbol{b'},\boldsymbol{r}_{<i}} + \theta_i + r_i \pi } \right) \nonumber \\
\otimes \ket{+_{-a_{i}^{''\text{  } \boldsymbol{b'},\boldsymbol{r}_{<i}} - r_i \pi }}\bra{+_{-a_{i}^{''\text{  } \boldsymbol{b'},\boldsymbol{r}_{<i}} - r_i \pi }} \nonumber  \end{split}\\
= \sum_{r_i} p(r_i) \frac{I_3}{2^3}\otimes \ket{+_{-a_{i}^{''\text{  } \boldsymbol{b'},\boldsymbol{r}_{<i}} - r_i \pi }}\bra{+_{-a_{i}^{''\text{  } \boldsymbol{b'},\boldsymbol{r}_{<i}} - r_i \pi }} \nonumber \\
 = \frac{I_4}{2^4} \nonumber
\end{gather}

where $I_n = \bigotimes_n I$. The last step was possible because each corrected computation angle $a''_{i}$ depends only on past $r$'s.

And finally (for $u=(t_g-1)m+1$),

\begin{gather}
\sum_{\theta _{u},r_u} p(\theta _{u},r_u) \ket{\delta_{u}^{\theta_u, r_u}}\bra{\delta_{u}^{\theta_u, r_u}} \otimes  \ket{+_{-a_{u}^{'} - r_u \pi }}\bra{+_{-a_{u}^{'} - r_u \pi }}  \nonumber \\
\begin{split}
= \sum_{r_u} p(r_u) \left( \sum_{\theta _{u}} p(\theta _{u})  \ket{a_{u}^{'} + \theta_u + r_u \pi }\bra{a_{i}^{'} + \theta_u + r_u \pi } \right) \nonumber \\
 \otimes  \ket{+_{-a_{u}^{'} - r_u \pi }}\bra{+_{-a_{u}^{'} - r_u \pi }} \nonumber \end{split} \\
= \frac{I_4}{2^4} \nonumber
\end{gather}

For $k \neq t_g$, if $km \geq t_k \geq km-n+1$ the above procedure will eventually give:

\begin{gather}
 \mathcal{P'}^{\dagger} (\frac{I_{4m-3n-1}}{2^{4m-3n-1}} \otimes  \ket{+_{\theta _{t_k}}}  \bra{+_{\theta _{t_k}}} )  \mathcal{P'}  \nonumber \\
= \frac{I_{4m-3n-1}}{2^{4m-3n-1}} \otimes  \ket{+_{\theta _{t_k}}}  \bra{+_{\theta _{t_k}}}  \nonumber
\end{gather}

If $km-n \geq t_k \geq (k-1)m+1$ the above procedure will eventually give:

\begin{gather}
\mathcal{P'}^{\dagger} (\frac{I_{4m-3n-4}}{2^{4m-3n-4}} \otimes \ket{\delta_{t_k}^{\nu_T}}\bra{\delta_{t_k}^{\nu_T}} \otimes \ket{+_{r_{t_k} \pi }} \bra{+_{r_{t_k} \pi }} )  \mathcal{P'}  \nonumber \\
= \frac{I_{4m-3n-4}}{2^{4m-3n-4}} \otimes \ket{\delta_{t_k}^{\nu_T}}\bra{\delta_{t_k}^{\nu_T}} \otimes \ket{+_{\theta _{t_k}}}  \bra{+_{\theta _{t_k}}}  \nonumber
\end{gather}

And for $k = t_g$ the result will be: $\bigotimes_{4m-3n} I$, which concludes the proof.

\end{proof}




\end{document}